\documentclass[aps,pre,twocolumn,superscriptaddress]{revtex4-2}

\usepackage{amsmath, amssymb, mathtools, hyperref, graphicx, enumitem, xcolor, amsthm}
\usepackage{tabularx}
\usepackage{overpic}

\def\eu{\ensuremath{\mathrm{e}}}

\def\du{\ensuremath{\mathrm{d}}}

\newtheorem{thm}{Theorem}
\newtheorem{lem}{Lemma}

\begin{document}

\title{Degree-targeted cascades in modular, degree-heterogeneous networks}

\author{Jordan Snyder}
\email{jsnyd@uw.edu}
\affiliation{Graduate Group in Applied Mathematics, University of California, Davis, CA 95616}
\affiliation{Department of Applied Mathematics, University of Washington, Seattle, WA 98195}
\author{Weiran Cai}
\affiliation{Department of Computer Science, University of California, Davis, CA 95616}
\author{Raissa M. D'Souza}
\affiliation{Department of Computer Science, University of California, Davis, CA 95616}
\affiliation{Department of Mechanical and Aerospace Engineering, University of California, Davis, CA 95616}
\affiliation{Santa Fe Institute, Santa Fe, NM 87501}

\begin{abstract}

The dynamics of cascading activation, such as rapid changes in public opinion and the outbreak of disease epidemics, have a crucial dependence on the connectivity patterns among the agents. We study cascading dynamics in modular, degree-heterogeneous networks, and consider the impact of intra-module seeding strategy on inter-module spread. Specifically, we establish that although activating the highest-degree nodes
is more effective than random selection at growing a cascade locally, there is a critical level of inter-module connectivity required for a cascade to cross from one module to another, irrespective of the seeding strategy. We present an analytical proof of this statement for the case that each module has the same degree distribution and all module pairs have the same inter-module connectivity, while our simulation results suggest its validity for more general situations, including a ring of modules. Interestingly, we find that on a network comprised of two modules, this critical level is primarily determined by the degree distribution of the \emph{alter} module, as opposed to the seed module.
Our analytical approach extends a method developed by Gleeson, but is able to capture different seeding strategies using only one dynamical variable per module, namely the conditional exposure probability. Our work shows that the possibility of a global cascade depends sensitively on inter-module connectivity, and less on the intra-module seeding strategy. This suggests, for example, that slight changes to inter-module connectivity can be a feasible intervention strategy to promote or inhibit global cascades.

\end{abstract}

\maketitle

\section{Introduction and Background}

Many important collective phenomena can be viewed as cascades of activation, and studies using network models help reveal general principles that relate the interaction patterns of individuals to the spread of activity among them. Prime examples of this kind of insight include the absence of an epidemic threshold for the SIS model of epidemic spreading on networks with sufficiently broad degree distributions \cite{Pastor-Satorras2001}, and the dependence of global cascades on underlying network characteristics \cite{Watts2002}. Here we focus especially on the impact of seeding strategy on the outcome of a cascade process in the presence of modularity and degree heterogeneity, since these properties are present in real-world networks and have repeatedly been identified as having a strong effect on the outcome of dynamics on networks \cite{Saavedra2011,May2008,Memmott2004,Burgos2007,Turalska2019}. For example, we may be interested in the vulnerability of a mutualistic ecosystem to cascades of extinction events \cite{Dunne2002,Sole2001,Morris2003}, or we may be interested in the efficacy of information diffusion in the case that seeding of information is targeted at well-connected individuals within a certain community \cite{Vicario2016,Morone2015,Iribarren2011}.

A network-theoretic treatment of cascades was established through the use of generating functions, as first shown by Watts \cite{Watts2002}, who formulated a simple network dynamical model that exhibits cascades whose sizes may either follow a power-law distribution or a bimodal distribution, depending on network properties. Generating functions had previously been used to study structural properties of networks by Callaway et al., who considered percolation (site, bond, and joint site-bond) with degree-dependent site occupation \cite{Callaway2000}. A now classic result they show is that networks with a sufficiently broad degree distribution lose large-scale connectivity rapidly upon removal of their highest-degree nodes, a fact that has been explored in depth by several other studies \cite{Cohen2001,Albert2000,Broder2011}.

With respect to cascading dynamics, Gleeson generalized the problem class of that posed by Watts, and developed a method for approximating the outcome. Gleeson's method is valid for irreversible binary-state dynamics, where the probability that any given node becomes active is a non-decreasing function of its number of active neighbors. Moreover the network on which the dynamics take place should be either a degree-corrected stochastic block model \cite{Karrer2011}, or a configuration model with arbitrary degree correlations \cite{Newman2003}. In the former case, the initial activation probability should be independent of degree.

For the more structured case of modular networks, the outcome of activation cascades has been studied by Nematzadeh et al. They specifically consider two modules and established conditions on the connectivity between them that allow an information cascade started in one module to spread to the other \cite{Nematzadeh2014}. Naturally there is a lower bound on the density of links connecting the two modules in order for a cascade to spread, but interestingly, there is also an upper bound. Intuitively the reason is that links between the modules make it more difficult for nodes in the first module (the one where the cascade is initialized) to be activated, because  links to the alter module are less likely to connect to an active node. They corroborate their results by direct simulation of the network dynamics and by the analytical approximation developed by Gleeson \cite{Gleeson2008}.

Note that another class of dynamics termed ``cascades'' are load-shedding or sandpile dynamics, as in the classic model by Bak-Tang-Wiesenfeld \cite{Bak1987,Bak1988}. In these dynamics, nodes hold units of load, and shed load to their neighbors when they reach capacity. Those neighbors may subsequently reach capacity, shedding more load, and a cascade of load shedding ensues. It has been observed that on modular, degree-heterogeneous networks, degree correlations can affect the ability of a cascade to spread from one module to another \cite{Turalska2019}. It has also been shown that for sandpile dynamics, there exists an optimal level of connectivity between modules that minimizes the probability of large local cascades \cite{Brummitt2012}. This is in contrast with the optimality found in linear threshold dynamics \cite{Nematzadeh2014}, where global cascades are facilitated by increasing connectivity between modules. This demonstrates the same network structure can be optimal for different objective functions given different dynamics.

We are interested to understand how local intra-module activation can influence inter-module spread. More specifically, we are concerned with understanding the effects of degree-targeted seeding on the dynamics of cascades on modular networks with heterogeneous degree distributions. We begin by specifying a modular, degree heterogeneous network ensemble, and dynamics that occur on those networks. We then outline an analytical method, extending the work by Gleeson \cite{Gleeson2008}, that lets us succinctly describe the dynamics at the level of modules, even when seeding probability depends on node degree. Finally we report results of both direct simulation of network dynamics and the analytical approximation, observing good agreement. We find that while seeding high-degree nodes facilitates cascades with fewer seed nodes, there is a critical level of module interconnectivity required for a cascade to spread between modules, which is independent of the seeding strategy.

{ While we focus initially on networks consisting of two modules with the same degree distribution, our analysis suggests that similar results hold in more general situations. In particular, we find evidence of a critical level of interconnectivity in a network comprised of a ring of modules, and in a two-module network where the modules have different degree distributions. Interestingly, results in the latter case demonstrate that the degree distribution of the alter module, rather than the seed module, is key in determining the critical level of interconnectivity.}

\section{Problem Statement}

\label{sec:problem_statement}

To make the problem concrete, we incorporate degree-targeted seeding into modular, degree-heterogeneous networks, building on the model studied by Nematzadeh et al. for studying the latter \cite{Nematzadeh2014}. The model consists of two-state linear threshold dynamics \cite{Granovetter1978,Watts2002} on a network with two modules, wherein the degree of each node is drawn from some prescribed degree distribution, and the fraction of intra- vs. inter-module links in the network is specified. Aside from these constraints, links are placed at random, i.e. without degree correlations; previous work on cascades in modular, degree-heterogeneous networks in the context of cascades of load redistribution has shown that degree correlations (both within and between modules) has an impact on the ability of cascades to spread from one module to another \cite{Turalska2019}, but this is beyond the scope of the present work, which is focused on monotonic threshold dynamics. By adjusting the degree distribution and the fraction of intra- vs. inter-module links, we can explore both the space of degree heterogeneity and that of modularity. We complete the specification of the dynamics by selecting a set of \emph{seed nodes} that are active initially.

We consider that all the seed nodes are contained within a single module, so that we can discern conditions under which a cascade can spread to the second module, and consider two different seeding strategies. First, we select a certain fraction of nodes uniformly at random, i.e. independently of their degree. Second, we select the same fraction of nodes but ensure that they are of the highest possible degree. We compare the two seeding strategies across the joint space of degree heterogeneity and modularity.
As expected, the number of seed nodes required for a global cascade is lower for targeting high degree nodes than for uniform seeding, but surprisingly we find that the critical fraction of inter-module links for the cascade to spread from one module to another is independent of the seeding protocol used. As discussed later, this can provide a feasible intervention strategy to enhance or suppress global cascades.

To develop a framework for analyzing the process, we apply a generalization of a modeling approach introduced by Gleeson \cite{Gleeson2008}. That approach applies to irreversible binary-state dynamics, where the probability that a node becomes active is a nondecreasing function of its number of active neighbors. This model class is quite general; it includes both fractional and absolute threshold models, as well as dynamics that compute site percolation and $k$-core sizes \cite{Gleeson2008}. A thorough and general description of this model class is given in Secs. \ref{sec:network_model} and \ref{sec:dynamics}.

{In Gleeson's framework, one assumes that the network can be approximated as a tree. Under the tree-like assumption, one can divide the network into levels that start at the leaves and end at the root, and one assumes that cascades of activation spread ``up'' the tree. This set of assumptions leads to a set of equations for the probability that a node with a given degree, located on a given tree-level, is activated. Remarkably, one can interpret the tree-level as a time variable and obtain highly accurate models for the time-course of cascade processes, as well as the eventual state of the network (i.e. total fraction of nodes that are active).}

Gleeson's technique applies directly when seed nodes are selected uniformly at random; here we demonstrate that his approach can be adapted to incorporate degree-targeted seeding. A direct application of Gleeson's technique would require one dynamical variable for each pair of (module, degree) values. In contrast, we show that the \emph{conditional exposure probability}
is independent of the nodes's degree and is sufficient 
to recover the total fraction of nodes in each module that are active. A similar analysis was carried out by Hackett in the case of a simple network without modular structure \cite{Hackett2011}; our derivation treats the more complicated case of modular network structures which allows us to establish the consequences of modular structure on cascade dynamics.


\subsection{Random network model}
\label{sec:network_model}
Here we describe the ensemble of networks that we consider. The networks in question are modular, and have arbitrary degree distributions in each module. For now we suppose, for simplicity, that all modules are of equal size; analogous formulations for unequally-sized modules are also possible.

Let $d$ denote the number of modules, and $n$ the number of nodes in each module. For each module $i=1,\dots,d$, let $p^{(i)}_k$ be the probability that a randomly selected node from module $i$ has degree $k$. Further, let $e\in [0,1]^{d\times d}$ be the \emph{mixing matrix} \cite{Newman2003a}, where $e_{ij}$ is the fraction of links leaving a module-$i$ node that end at a module-$j$ node. We take our network to be undirected, so $e_{ij} = e_{ji}$.

To produce a network that conforms to the above specifications, we can perform a stub-matching procedure, as is standard practice to sample from the configuration model \cite{Molloy1995}, where a stub is a half-edge attached to a node.
The networks we produce in this way are essentially the same as those that conform to the degree-corrected stochastic blockmodel, introduced by Karrer and Newman \cite{Karrer2011}, although their presentation is in terms of edge existence probabilities rather than stub matching. We use stub matching here as it leads to an algorithm with faster runtime and smaller memory requirements for large networks with asymptotically constant mean degree.
Since here we impose a nontrivial modular structure, implementation details are given in Appendix \ref{sec:stub-matching}.

\subsection{Dynamics} \label{sec:dynamics}

Here we describe the class of dynamics that we consider, namely monotonic threshold dynamics.

Let $N$ be the total number of nodes in a network and let $A\in \{0,1\}^{N\times N}$ denote its adjacency matrix: $A_{ij}=1$ if and only if nodes $i$ and $j$ share an edge, and we assume all edges are undirected and unweighted. Let $k_i = \sum_j A_{ij}$ denote the degree of node $i$.

Let $u\in \{0,1\}^N$ denote the \emph{state} of the system; $u_i=1$ if node $i$ is active, and zero otherwise. Following Watts \cite{Watts2002}, we introduce a dynamic on $u$ according to the rule
\begin{equation}
u_i(t+1) = \begin{cases}
1 & \sum_j A_{ij}u_j(t) > \theta k_i \text { or } u_i(t)=1\\
0 & \text{else}
\end{cases}
\end{equation}
where $\theta\in[0,1]$ is the \emph{threshold}. In words, a node becomes active if at least a fraction $\theta$ of its neighbors are active, and an active node remains active for all time.

A few properties to note about these dynamics:
\begin{enumerate}
	\item Most importantly, there is an absorbing state, in contrast with other work analyzing binary state approximations of population dynamics in mutualistic ecosystems \cite{Campbell2010,Campbell2012}. This also rules out the possibility of compensatory perturbations \cite{Motter2004,Cornelius2013,Sahasrabudhe2011}, because activating more nodes will never decrease the number of nodes that are eventually activated.
	
	\item We consider the limit where all nodes update their state at the same time. 
	Because the dynamics are monotonic (i.e. an active node never becomes inactive), the eventual state of the system is independent of the order in which nodes are updated. Therefore we choose to update all nodes in synchrony and focus on the eventual steady state and not on the dynamics leading towards it.
	
	\item The threshold $\theta$ is here assumed to be the same for all nodes. In general $\theta$ may vary from node to node, and can be assigned at random.
\end{enumerate}
This dynamic can happen on any network; we will focus on the modular and degree-heterogeneous networks described in Sec. \ref{sec:network_model}.

\section{Methods}
\label{sec:analytics}

In this section we show how to derive equations governing the fraction of active nodes in each module, adapting results of \cite{Gleeson2008} to allow for degree-targeted seeding with minimal increase in computational burden.

\subsection{Treelike approximation with degree-dependent seeding}
\label{sec:degree_targeted_seeding}

We now show how to extend Gleeson's framework to allow for seed nodes to be selected with a probability that depends on their degree. While degree-targeted seeding is possible already in Gleeson's formulation, a direct application would require keeping track of a dynamical variable for each degree class. In contrast, we find that equivalent results are possible using only one dynamical variable per module, as shown below.

By degree-targeted seeding, we mean that the probability that a node is active initially depends both on its module membership and its degree, and we write $\rho^{(i)}_{0, k}$ as the probability that a randomly selected node in module $i$ having degree $k$ is active initially. This allows us to model analytically cases where, for example, a cascade is initialized by activating the highest-degree nodes in a single module. The case that initialization is independent of node degree is recovered by setting $\rho^{(i)}_{0, k} = \rho^{(i)}_{0}$ for all $k$. The total fraction of nodes initially activated is given by
\begin{equation}
\rho^{(i)}_{0,\text{tot}} = \sum_k p^{(i)}_k \rho^{(i)}_{0,k}.
\end{equation}

Following \cite{Gleeson2008,Gleeson2007}, we assume the network can be modeled as an infinite tree, with levels indexed starting with the leaves at the bottom (level $n = 0$) and the root infinitely high (level $n\to \infty$). We then formulate a recursion relation that describes updating nodes' states up the levels of the tree (from children to parents), assuming at every stage that all nodes at lower levels have already been updated. To do this, let $q^{(i)}_{n,k}$ be the probability that a node in module $i$ having degree $k$ at level $n$ of the tree is active, conditional on its parent (at level $n+1$) being inactive. Thanks to this conditioning, we can write a recursion relation for $q^{(i)}_{n,k}$ directly.

Accounting for degree-targeted seeding, we have
\begin{align}
q^{(i)}_{n+1,k} =& \rho_{0,k}^{(i)} + (1- \rho_{0,k}^{(i)}) \sum_{m=0}^{k-1} \binom{k-1}{m} \left(\overline{q}^{(i)}_n\right)^m\nonumber\\
 &\times \left(1-\overline{q}^{(i)}_n\right)^{k-1-m} F^{(i)}(m,k)
\label{eq:q_update_degree_weighted}
\end{align}
where $F^{(i)}(m,k)$ is the probability that a node in module $i$ having degree $k$ and $m$ active neighbors becomes active, and $\overline{q}^{(i)}_n$, which we call the \emph{conditional exposure probability}, is the probability that a level-$n$ child of an inactive module-$i$ parent is active. It is given by
\begin{equation}
\overline{q}^{(i)}_n = \frac{1}{\sum_j e_{ij}} \sum_j e_{ij} \left[\sum_k \frac{k}{z^{(j)}} p^{(j)}_k q^{(j)}_{n,k}\right],
\label{eq:qbar_degree_weighted}
\end{equation}
where $z^{(j)} \coloneqq \sum_k k p^{(j)}_k$ is the average degree of nodes in module $j$. The total density of active nodes in each module $i$ is given by
\begin{align}
\rho^{(i)}_{n+1} &=&& \sum_k p^{(i)}_k \biggl[\rho_{0,k}^{(i)} + (1-\rho_{0,k}^{(i)}) \nonumber\\
&\,&&\times \sum_{m=0}^{k}\binom{k}{m} \left(\overline{q}^{(i)}_n\right)^m \left(1-\overline{q}^{(i)}_n\right)^{k-m}F^{(i)}(m,k)\biggr]\nonumber\\
&\coloneqq&& H^{(i)}(\overline{q}^{(i)}_n)
\label{eq:rho_update_degree_weighted}
\end{align}

The conditional exposure probabilities $\overline{q}^{(i)}_n$ represent the probability that a randomly chosen level-$n$ neighbor of a node in module $i$ is active, hence the edge-following degree distribution $kp^{(j)}_k/z^{(j)}$ in Eq. (\ref{eq:qbar_degree_weighted}). With this in mind, we can interpret the update rule Eq. (\ref{eq:q_update_degree_weighted}) as follows. Given a node in module $i$ with degree $k$ at tree level $n$ whose parent is not active, it was active from the start with probability $\rho_{0,k}^{(i)}$. With probability $1-\rho_{0,k}^{(i)}$, it was not active from the start, and has a chance to be activated by its $k-1$ children. Since we assume the network is treelike, each of those $k-1$ children is active independently with probability $\overline{q}^{(i)}_n$, meaning that the probability that $m$ of them are active at once is $\binom{k-1}{m} \left(\overline{q}^{(i)}_n\right)^m \left(1-\overline{q}^{(i)}_n\right)^{k-1-m}$. If $m$ children are active, then the focal node becomes active with probability denoted by $F^{(i)}(m,k)$.

Next, consider the unconditional probability $\rho_{n+1}^{(i)}$ that a node in module $i$ at tree level $n+1$ is active. If we choose a node at random from module $i$, it has degree $k$ with probability $p^{(i)}_k$. Given that it has degree $k$, it was active from the start with probability $\rho_{0,k}^{(i)}$. Otherwise, it has $k$ neighbors, each of which is active with probability $\overline{q}^{(i)}_n$. The total number of active neighbors then follows the binomial distribution seen in Eq. (\ref{eq:rho_update_degree_weighted}), and the focal node becomes active with probability $F^{(i)}(m,k)$ if it has $m$ active neighbors.


Finally, we show that the conditional exposure probability obeys a closed recursion relation with respect to $n$, and that the dependence on degree $k$ enters in only through the static parameters $p^{(j)}_k$ and $\rho^{(i)}_{0,k}$. Composing Eq. (\ref{eq:qbar_degree_weighted}) with Eq. (\ref{eq:q_update_degree_weighted}), we have
\begin{align}
&\overline{q}^{(i)}_{n+1} =\frac{1}{\sum_j e_{ij}} \sum_j e_{ij} \bigg[\sum_k \frac{k}{z^{(j)}} p^{(j)}_k \bigg(\rho_{0,k}^{(j)} + (1- \rho_{0,k}^{(j)}) \nonumber \\
& \times\sum_{m=0}^{k-1} \binom{k-1}{m} \left(\overline{q}^{(j)}_n\right)^m \left(1-\overline{q}^{(j)}_n\right)^{k-1-m} F^{(j)}(m,k)\bigg)\bigg] \nonumber  \\
&\coloneqq G^{(i)}(\overline{q}_n)
 \label{eq:qbar_update_closed_degree_weighted}
\end{align}

Note that if $\rho^{(i)}_{0,k} = \rho^{(i)}_0$ for all $k$, then Eq. (\ref{eq:qbar_update_closed_degree_weighted}) is exactly equivalent to the formulation in \cite{Gleeson2008}.

We can further describe the trajectory of the dynamics obtained by assuming that nodes update asynchronously at a rate of $f$ per unit time. Following the same reasoning as in \cite{Gleeson2008}, we have
\begin{equation}
\label{eq:ODE_degree_weighted}
\begin{split}
\frac{\du\overline{q}^{(i)}(t)}{\du t} &= f \left[G^{(i)}(\overline{q}(t)) - \overline{q}^{(i)}(t)\right]^+\\
\frac{\du{\rho}^{(i)}(t)}{\du t} &= f \left[H^{(i)}(\overline{q}^{(i)}(t)) - \rho^{(i)}(t)\right]^+
\end{split}
\end{equation}
where $[\cdot]^+$ denotes the positive part.


\section{Results}

\subsection{Model parameter specification}
\label{sec:optimal_modularity}
Here we describe extensions to the research presented in \cite{Nematzadeh2014}, on existence of optimal levels of interconnectivity in modular networks. In that work, the authors consider a network composed of two modules, with a fraction $\mu\in[0,1]$ of the links joining nodes in the same module and the remaining $(1-\mu)$ fraction of the links joining nodes in different modules. This corresponds to choosing the mixing matrix to be
\begin{equation}
e = \begin{bmatrix}
1-\mu & \mu \\
\mu & 1-\mu
\end{bmatrix}.
\end{equation}
On this network, they consider linear threshold dynamics of the type described above, with seed nodes localized to one of the modules.

The main result of \cite{Nematzadeh2014} is that for certain values of the seed density, there is an optimal range of values of the modularity parameter $\mu$, such that within this range the cascade covers the whole system, while on either side of this range the cascade remains localized to the module where it began. Their conclusions follow from both direct simulations of dynamics on networks and calculations based on the analytical framework developed by Gleeson \cite{Gleeson2008}.

We now generalize their results in two ways. First, we consider degree distributions with a tunable extent of degree heterogeneity, and second, we allow for degree-targeted seeding. Sec. \ref{sec:general_network_structures} and Appendices \ref{sec:monotonicity-appendix},\ref{sec:well-ordering-appendix}, and \ref{sec:next-nearest-neighbor-appdx} consider the possibility of more than two modules and different degree distributions in each module.

The first aspect, degree heterogeneity, was discussed in the SI of \cite{Nematzadeh2014}, where the authors present results for the LFR benchmark networks \cite{Lancichinetti2008}, and state that degree heterogeneity does not change the results qualitatively. Here, instead, we treat degree heterogeneity explicitly as a control parameter, quantified by $p_\text{nest}\in[0,1]$, so named as a reference to the concept of \emph{nestedness} in theoretical ecology \cite{Atmar1993,Bascompte2003}, which has been found to be largely explained by degree heterogeneity \cite{Jonhson2013}. The parameter $p_\text{nest}$ enters the analysis through the degree distribution:
\begin{equation}
p_k = p_\text{nest} p^\text{pow}_k + (1-p_\text{nest})p^\text{poi}_k
\end{equation}
where $p^\text{pow}$ and $p^\text{poi}$ are power law and Poisson degree distributions, respectively, each with mean $z$ (which we take, for now, to be 20). For completeness, we have
\begin{equation}
p^\text{pow}_k = {1 \over \zeta(\gamma, \lambda)} (\lambda + k )^{-\gamma}
\end{equation}
with $\lambda$ chosen such that $\sum_k k p^\text{pow}_k = z$, and $\zeta(\gamma,\lambda) = \sum_k (\lambda + k)^{-\gamma}$ is the Hurwitz zeta function, and
\begin{equation}
p^\text{poi}_k = \frac{z^k \eu^{-z}}{k!}.
\end{equation}
{We note that the above degree distribution puts nonzero weight on $k=0$, and that it is not guaranteed that random networks with this degree distribution will be connected. In numerical experiments, we circumvent this problem by selecting the largest connected component, which for the parameter values we consider is typically nearly the entire network. Throughout, we use $\gamma = 2.5$ for the power law exponent.}
%

The question of degree-targeted seeding, however, was not addressed in \cite{Nematzadeh2014}, and we use the new analytic framework presented in Sec. \ref{sec:degree_targeted_seeding} to do so. Specifically, we compare two different seeding strategies: max-degree and uniform. Max-degree seeding (described in detail in Appendix \ref{sec:Validation}) selects the highest-degree nodes present in a network, while uniform seeding selects nodes with equal probability. Based on previous work on percolation \cite{Albert2000,Cohen2000,Cohen2001}, we anticipate that max-degree seeding will produce different outcomes than uniform seeding, and that this difference will increase with increased degree heterogeneity.

Yet, from both numerical simulations and analytical derivations, we find however that the modular structure of the network, in terms of the mixing matrix, places a limit on the possible extent of an inter-module cascade when initiated in a single module, independent of the seeding protocol.

\subsection{Impact of seeding on cascade size}

An overview of our results is presented in Fig. \ref{fig:optimal_modularity_heatmaps}; the top two rows were generated using Eq. (\ref{eq:rho_update_degree_weighted})-(\ref{eq:qbar_update_closed_degree_weighted}) while the bottom two rows are results from direct simulation of the network dynamics. The second and fourth rows (uniform seeding) exactly recreate the results of \cite{Nematzadeh2014}. As expected, the system becomes more resilient to cascading failures as the degree heterogeneity increases when seeding is independent of degree, while the opposite tendency is observed when max-degree seeding is used. This spectrum of behavior aligns well with the famous ``robust yet fragile" nature of random graphs with power law degree distributions, which are robust to random attacks and fragile to targeted attacks\cite{Doyle2005,Cohen2000,Cohen2001,Callaway2000,Albert2000}.


\begin{figure*}
\includegraphics[width=\textwidth]{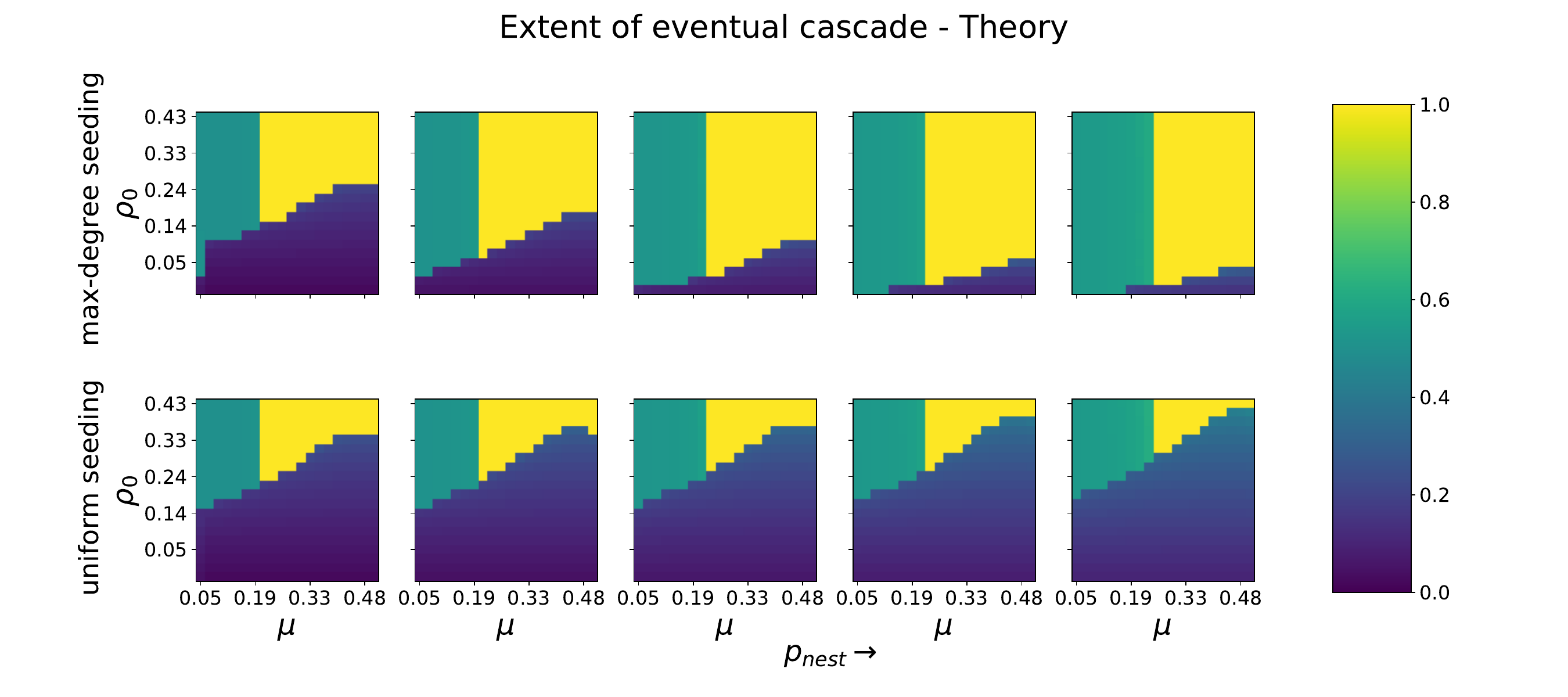}
\includegraphics[width=\textwidth]{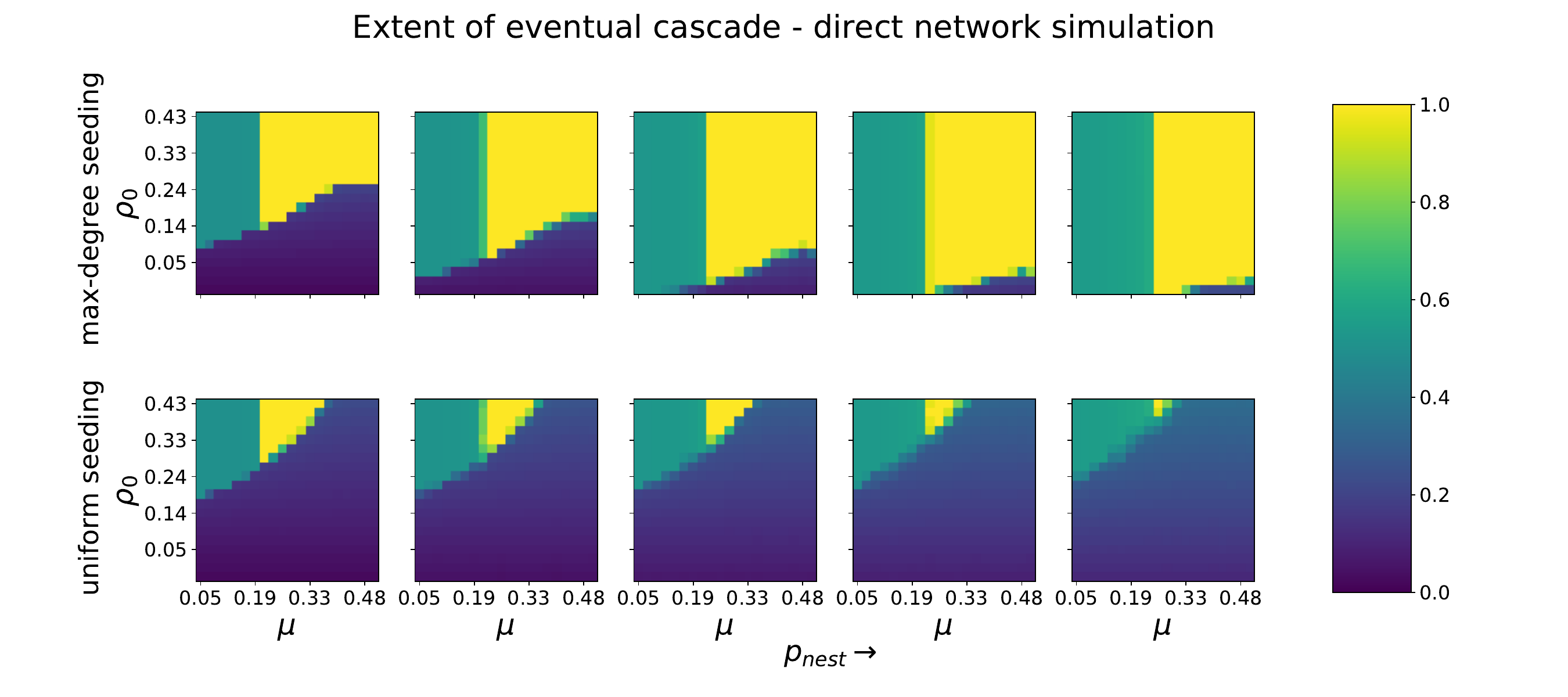}
\caption{Summary of the joint effect of inter-module connections ($\mu$) and degree heterogeneity ($p_\text{nest}$) on the extent of cascade spreading, under both uniform and degree-targeted seeding, as predicted by an analytic treatment, Eq. (\ref{eq:rho_update_degree_weighted})--(\ref{eq:qbar_update_closed_degree_weighted}) (top), and by direct simulation of the network dynamics, on networks of size $N = 2.5\times 10^4$, averaged over ten realizations (bottom). Color indicates extent of the eventual cascade, as a fraction of the whole network; yellow is 1, green is 0.5, and blue ranges between 0 and 0.3. The green region, present in every panel for $\mu \lesssim 0.2$, corresponds to the cascade completely covering the first module (where the seed nodes are located) and not spreading at all to the second. The blue region corresponds to the situation that the cascade spreads to only part of the first module. Parameters used here are: $\theta =0.4$, $\overline{k} = 20$. In these figures, $\rho_0$ is the fraction of the first module that is infected, which is off from the notation in \cite{Nematzadeh2014} by a factor of two. As expected, max-degree seeding leads to global cascades with fewer seed nodes than required by uniform seeding. Yet, the vertical line separating the green from yellow regions shows that there is a minimal value of inter-module connectivity $\mu$ required to allow a global cascade, regardless of how seed nodes are selected. \label{fig:optimal_modularity_heatmaps}}
\end{figure*}


However, both simulation and theory show an unexpected result that there is a critical level of interconnectivity required for inter-module cascades, which is independent of the seeding strategy. To see this, notice that a feature shared in common between the uniform and degree-targeted cases is the left boundary of the region corresponding to total activation (see Fig. \ref{fig:optimal_modularity_heatmaps}). That boundary corresponds to the critical value of $\mu$, the inter-module connectivity, required to allow complete activation of the second module.

The essential reason that this boundary is at the same location for both seeding protocols is that for $\mu$ smaller than $\mu_c$, \emph{no} seeding protocol restricted to the first module -- even the one with all nodes activated -- is able to initiate a cascade that fully activates the second. This observation immediately suggests a way to find $\mu_c$, namely the smallest $\mu$ such that having all nodes in the first module active initially leads to activation of the entire network (see Fig. \ref{fig:iterated_map_visualization}). This fact relies crucially on the fact that both modules have the same degree distribution, and that within-module edge density is greater than between-module edge density.

Now we make the above argument precise. We present the logical flow here and refer to the appendices for supporting calculations, {including a generalization to any number $d$ of modules and an arbitrary monotonic activation function.} All statements in this section are to be understood in the context of the reduced dynamics in the thermodynamic limit (i.e. Eq. (\ref{eq:qbar_update_closed_degree_weighted}), for $\overline{q}_n^{(i)}$ the conditional exposure probabilities).

First, we claim that if $\mu<0.5$ and more nodes in module 1 are active than in module 2, then the same will be true after updating all nodes' states. What we mean is that if $\mu<0.5$ and $\overline{q}^{(1)}_n > \overline{q}^{(2)}_n$, then $\overline{q}^{(1)}_{n+1} > \overline{q}^{(2)}_{n+1}$; see Appendix \ref{sec:well-ordering-appendix} for a proof. Clearly, then, if $\overline{q}^{(1)}_0 > \overline{q}^{(2)}_0$ and $\overline{q}^{(2)}_n\to1$ as $n\to\infty$, then $\overline{q}^{(1)}_n\to 1$ as $n\to \infty$ also. In other words, if module 2 is eventually completely activated, then module 1 will also be eventually completely activated. The necessary assumptions for that statement to be true are that the inter-module connectivity $\mu$ is less than $0.5$, both modules have the same degree distribution, and the initially active nodes are only in module 1.

The observation that if $\overline{q}^{(2)}_n \to 1$ then $\overline{q}^{(1)}_n\to 1$ allows us to derive a simpler condition for the eventual complete activation of module 2, and moreover shows that this condition is independent of the seeding protocol. Conceptually, the simplification captures the fact that the eventual state of the system is independent of the order in which nodes' states are updated. This fact rests on the fact that $G$, defined in Eq. (\ref{eq:qbar_update_closed_degree_weighted}, respects the partial order $\le$ on $[0,1]^d$. That is, $\overline{q} \le \overline{r} \implies G(\overline{q}) \le G(\overline{r})$ for all $\overline{q}\in [0,1]^d$; see Appendix \ref{sec:monotonicity-appendix} for proof.

That $G$ is order-preserving immediately implies a ``sandwich'' property, as follows. Suppose that if there are two sequences, $(\overline{q}_n)$ and $(\overline{r}_n)$ defined by $\overline{q}_{n+1} = G(\overline{q}_n)$ (and likewise for $\overline{r}$). First, each of these sequences has a limit, which we denote $\overline{q}_\infty$ and $\overline{r}_\infty$ respectively, and these limits satisfy $G(\overline{q}_\infty) = \overline{q}_\infty$ and $G(\overline{r}_\infty) = \overline{r}_\infty$. Suppose further that $\overline{q}_0 \le \overline{r}_0 \le \overline{q}_\infty$. Then $\overline{q}_n \le \overline{r}_n \le \overline{q}_\infty$ for all $n$, and since $\overline{q}_n\to \overline{q}_\infty$, we conclude $\overline{q}_\infty = \overline{r}_\infty$.

Finally we apply the above reasoning in the following way: Let $u$ be an initial condition in which only nodes in module 1 are active, and assume that the eventual state $u_\infty$ is the one in which all nodes are active (i.e. the vector of all ones). Then let $v$ be the vector of all ones in the first module and all zeros in the second module. Clearly, $u\le v \le u_\infty$, so by the theorem just stated, $v_n \to u_\infty$. In other words: if an initial condition localized to the first module would lead to complete activation of the whole network, then so would the initial condition in which all of the first module is active and none of the second is active.

To summarize, we conclude that in order for a cascade to activate the whole of the second module, it must be the case that the cascade activates all of the first module.
In particular, the fact of whether or not the whole second module is eventually active is independent of the details of the seeding protocol used in the first module, so long as the seeding protocol is sufficient to activate the whole first module.

To derive a condition on $\mu$ that allows a cascade to fully activate the second module, we write down the dynamics of $\overline{q}^{(2)}_n$ under the assumption that all nodes in module 1 are initially active. We have
\begin{align}
\overline{q}^{(2)}_{n+1} &= \mu + (1-\mu)\sum_k \frac{k}{z}p_k \nonumber\\
&\times\sum_{m=\lfloor \theta k \rfloor + 1}^{k-1}\binom{k-1}{m} \left(\overline{q}^{(2)}_n\right)^m \left(1-\overline{q}^{(2)}_n\right)^{k-1-m}
\label{eq:module_2_recursion}
\end{align}

Note crucially that Eq. (\ref{eq:module_2_recursion}) does not depend on the seeding protocol used, because we have replaced the initial condition with one in which all of module 1 is active. Therefore the left-boundary of the yellow region in Fig. \ref{fig:optimal_modularity_heatmaps} is located at the same value of $\mu$ regardless of seeding protocol.

Because Eq. (\ref{eq:module_2_recursion}) is an iterated map in one variable, we can visualize its behavior by plotting the right-hand side and inspecting its intersections with the line $y=x$. We present in Fig. \ref{fig:iterated_map_visualization} such a visual for two different values of $\mu$ and the same Poisson degree distribution with mean $\lambda = 20$.

\begin{figure}
\includegraphics[width=\columnwidth]{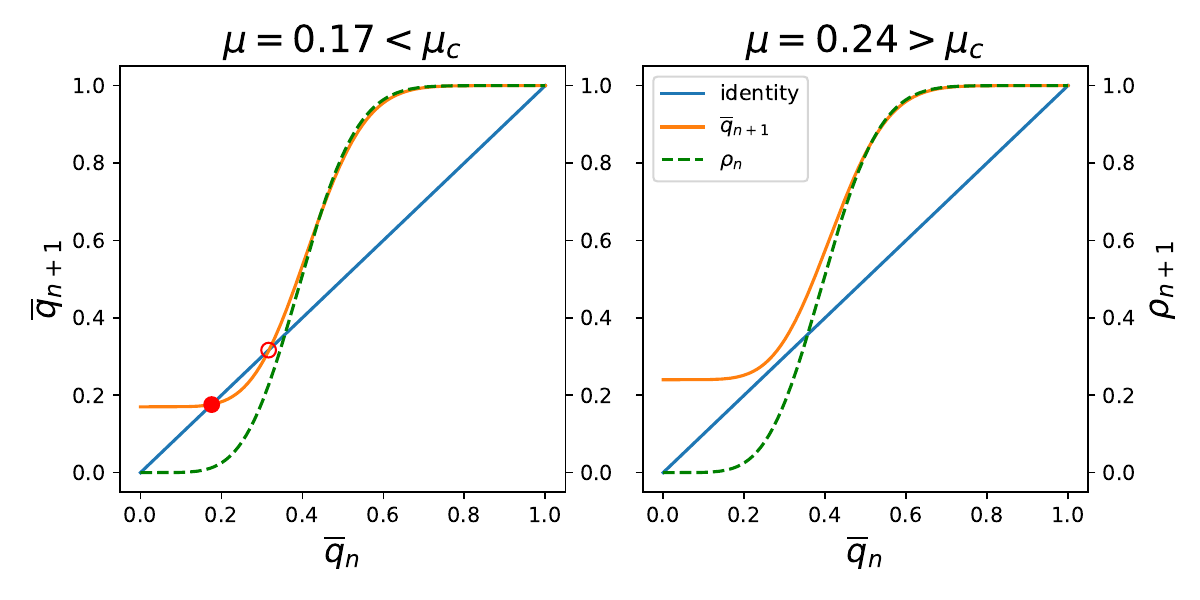}
\caption{Visualization of the iterated map Eq. (\ref{eq:module_2_recursion} for two different values of $\mu$, both with the same degree distribution, a Poisson with mean $\lambda = 20$. Orange curve is the right-hand side of Eq. (\ref{eq:module_2_recursion} while the blue line is the identity, so intersections correspond to fixed points. Notice that for $\mu = 0.24$ there is only one fixed point, at $\overline{q} = 1$, while for $\mu = 0.17$, an additional pair of fixed points (one stable, one unstable, at the solid and open red circles respectively) appears with $\overline{q}<1$, corresponding to partial activation of module 2. The green curve visualizes $\rho$, the fraction of active nodes corresponding to a given value of $\overline{q}$. \label{fig:iterated_map_visualization}}
\end{figure}

We can further compare the solution for $\mu_c$ based on Eq. (\ref{eq:module_2_recursion}) to dynamics on sampled networks. The results are shown in Fig. \ref{fig:mu_c_vs_pnest}. We observe qualitative agreement between the theory and sampled networks in the trend of $\mu_c$ with respect to $p_\text{nest}$, although there is a systematic difference between the prediction of the theory and the interval estimated from network simulations. Still, the intervals resulting from uniform and max-degree seeding of networks show no significant difference, supporting the conclusion that cascade hopping is independent of seeding protocol.

{Finally, we reiterate that the theory outlined here is valid for any number of modules, provided that they have the same degree distribution and that inter-module edge density is the same for all module pairs.}

\begin{figure}
\includegraphics[width=\columnwidth]{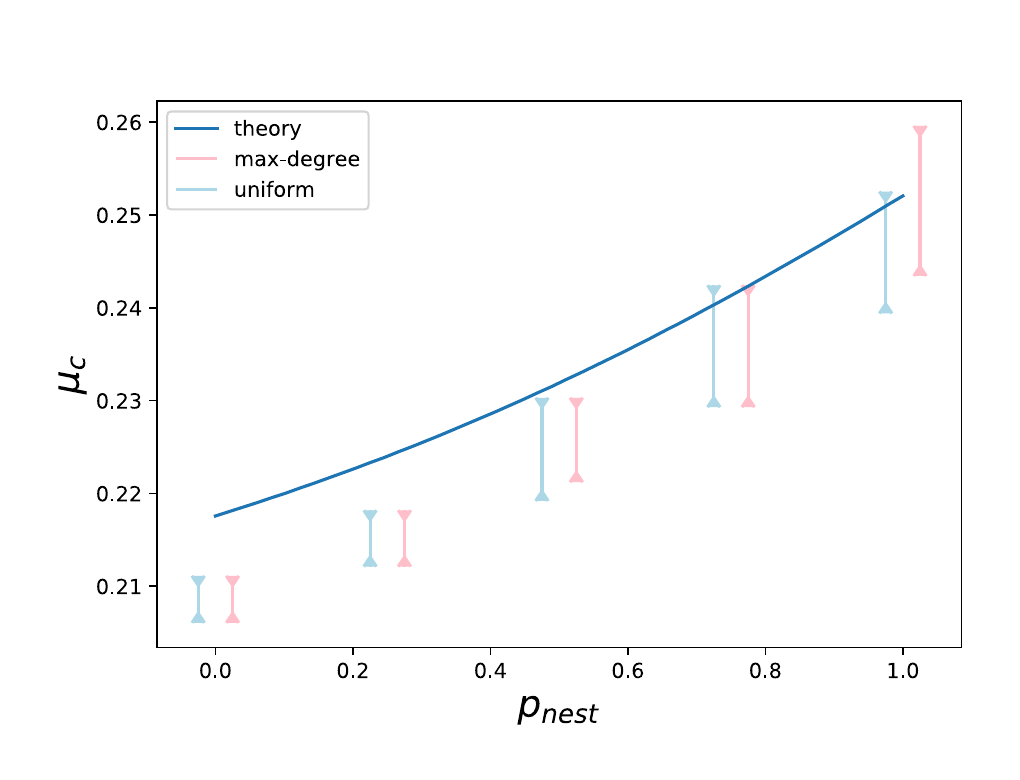}\\
\caption{Comparison of the critical value of $\mu$ required to achieve global activation as predicted by Eq. (\ref{eq:module_2_recursion}) vs. by direct simulation of the network dynamics. Theory curve was obtained via bisection search with respect to $\mu$, for each of 50 values of $p_\text{nest}$ evenly spaced between zero and one. For network simulations, we sampled five values of $p_\text{nest}$ evenly spaced from 0 to 1 (inclusive), and 100 values of $\mu$ evenly spaced between 0.2 and 0.3. For each pair of ($\mu$,$p_\text{nest}$) values, we sampled 10 networks, each of size $N=2.5\times 10^4$. For each network, we select a fraction $\rho_0$ of the nodes to be active initially, either according to maximum degree or uniformly at random. We repeat the experiment for 20 values of $\rho_0$, evenly spaced from 0.05 to 0.5. We then run the dynamics to equilibrium and record the eventual fraction of nodes in the network that are active. We then estimate a bracket around the true value of $\mu_c$ by plotting the largest $\mu$ such that none of the experiments show global activation, and the smallest $\mu$ such that all experiments do show global activation. Note the qualitative agreement with the trend of the theory curve, as well as close agreement between uniform and max-degree seeding. \label{fig:mu_c_vs_pnest}}
\end{figure}

Overall, the results from direct network simulation agree, at least qualitatively, very well with those from the theory, Eq. (\ref{eq:rho_update_degree_weighted})-(\ref{eq:qbar_update_closed_degree_weighted}). There is, however, a systematic difference between theory and direct simulation in the area of parameter space supporting global cascades, as shown in Fig. \ref{fig:mu_c_vs_pnest}. As the theory supposes that the networks are infinite and tree-like, we expect that this discrepancy is due to the presence of clustering (i.e. existence of short loops) in the finite networks we studied. While Melnik et al. showed that theories assuming a tree-like structure can nonetheless be ``unreasonably effective'' for finite networks, they also remark that the Watts threshold model with constant threshold (i.e. the dynamics we consider here) is especially sensitive to the presence of clustering \cite{Melnik2011}.

{We conclude with a heuristic explanation for the direction of the discrepancy between $\mu_c$ given by theory and that observed in real networks. Under the tree-like assumption, there is at most one path connecting any given seed node to any given non-seed node, along which the cascade of activation can proceed. In the actual network, there can be more than one such path, meaning there are multiple paths via which a single seed node can activate a given non-seed node. Since the dynamics we consider are threshold-based, repeated exposure increases the likelihood of activation. Therefore cascades of activation should spread more easily in real networks than the tree-like theory predicts, and fewer inter-module links should be necessary to cause a global cascade (i.e. $\mu_c$ should be lower) in real networks.}

\subsection{More general network structures}
\label{sec:general_network_structures}

The foregoing analysis was performed in a restricted setting, namely two-module networks where both modules have the same degree distribution. Below we present results for more general network structures.

First, we consider the case where the network is composed of two modules of equal size, but with different degree distributions. Specifically, we consider that one module has a Poisson degree distribution while the other has a power law degree distribution, and again investigate the ability of an initial activation localized to one module to spread to the other. The results are summarized in Fig. \ref{fig:heatmaps_different_distributions}.

\begin{figure}
\begin{overpic}[width=\columnwidth]{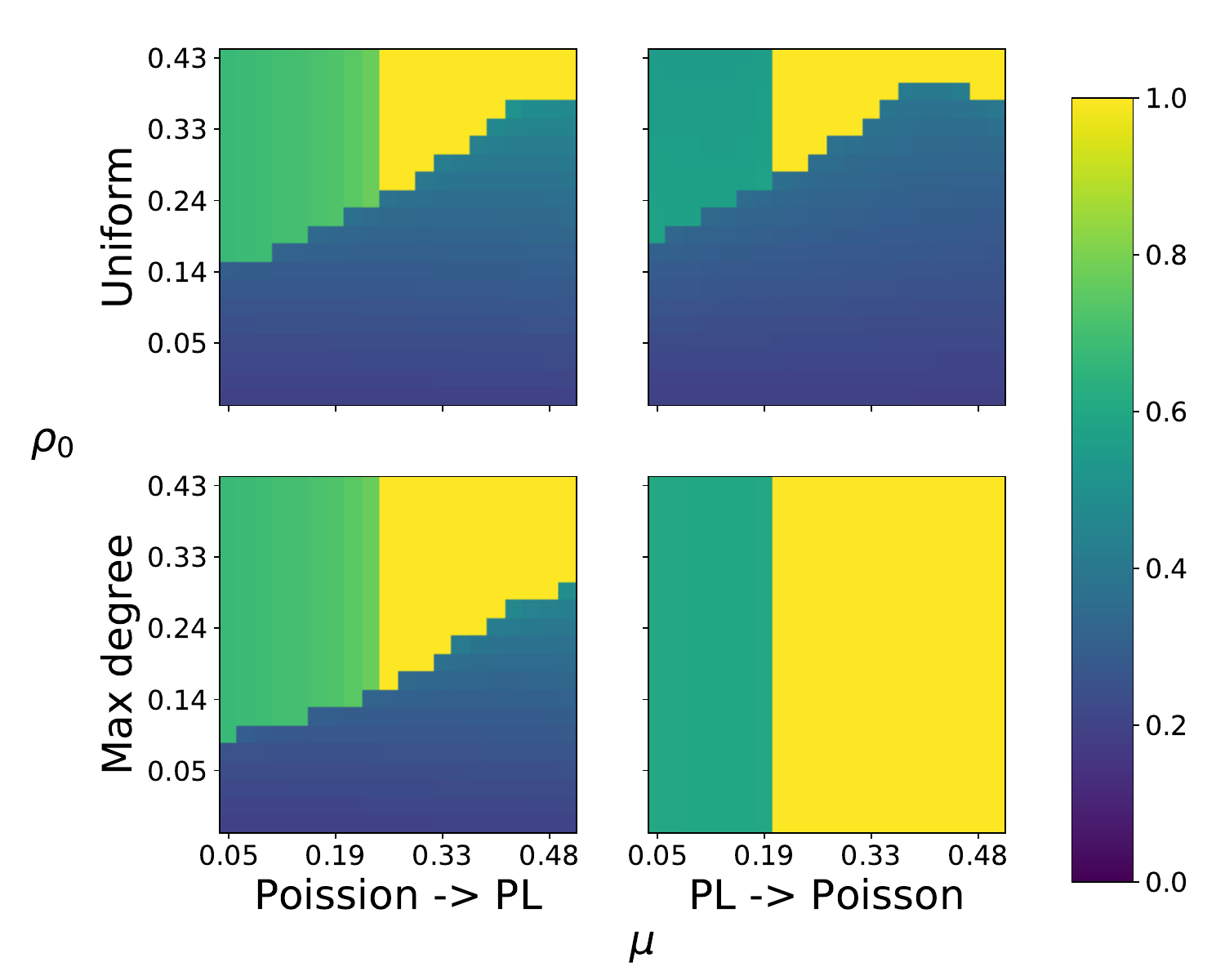}
	\put(43, 80){Theory}
\end{overpic}\\
\vspace{5mm}
\begin{overpic}[width=\columnwidth]{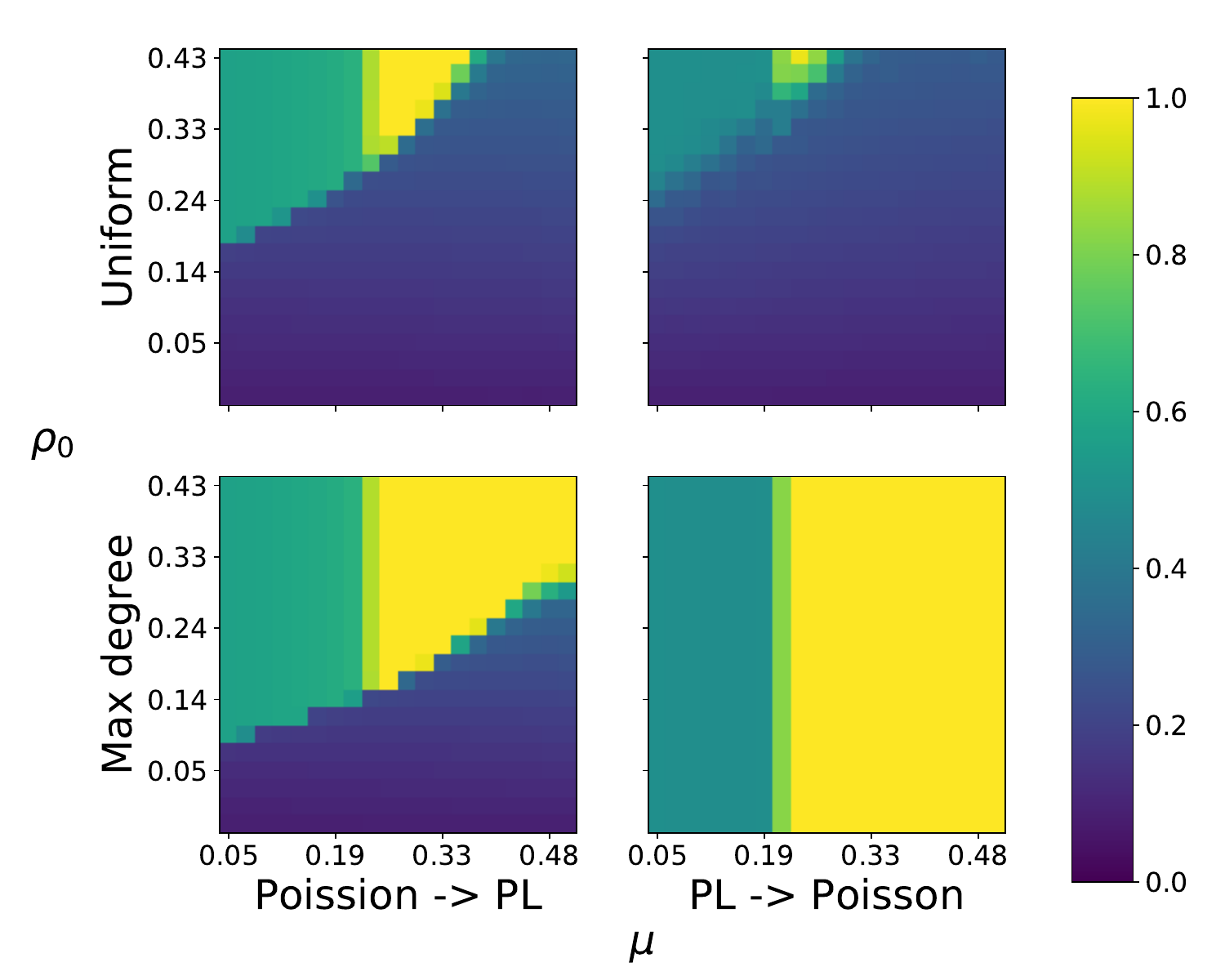}
	\put(35, 80){Network Simulation}
\end{overpic}
\caption{Summary of inter-module cascade spreading between modules with different degree distributions, as determined by direct network simulation (bottom) and analytical treatment (top). In the left column, seed nodes are selected from the module with Poisson degree distribution, and in the right column, seed nodes are selected from the module with the power law degree distribution. The axes are $\mu$, the fraction of inter-module links, and $\rho_0$, the fraction of seed nodes, and color indicates the extent of the eventual cascade, just as in Fig. \ref{fig:optimal_modularity_heatmaps}. \label{fig:heatmaps_different_distributions}}
\end{figure}

Overall, we observe that the impact of degree-targeted seeding is modest when seed nodes are selected from the module with a Poisson degree distribution (left hand column), but dramatic when seed nodes are selected from the module with a power law degree distribution (right hand column). In fact, we can see by comparing the heatmaps in Fig. \ref{fig:heatmaps_different_distributions} to those in Fig. \ref{fig:optimal_modularity_heatmaps} (with $p_\text{nest}= 0$ or $1$) that the degree distribution of the non-seed module appears to have little impact on the eventual cascade size.

Interestingly, we observe the existence of a critical level of interconnectivity required to support a global cascade, just as in the case where the two modules have the same degree distribution. Further, we observe that of the two cases shown in Fig. \ref{fig:heatmaps_different_distributions} (i.e. power-law$\rightarrow$Poisson and Poisson$\rightarrow$power-law), the value of $\mu_c$ is larger in the case that the seeded module has a Poisson degree distribution ($\mu_c \approx 0.23$), and smaller in the case that the seeded module has a power law degree distribution ($\mu_c \approx 0.19$). Meanwhile, Fig. \ref{fig:mu_c_vs_pnest} reveals that in the case of two identical modules, $\mu_c$ increases with degree-heterogeneity. In other words, in the case of different distributions, $\mu_c$ is smaller when the more heterogeneous network is seeded, but for identical modules $\mu_c$ increases with heterogeneity. We therefore conclude that the degree distribution of the \emph{alter} module, rather than the seed module, is key in determining $\mu_c$.


We now investigate the case where the network is composed of more than two modules, but with identical degree distributions. Specifically, we consider that each module has a power law degree distribution with exponent $\gamma = 2.5$, and suppose that the modules are connected in a ring topology. Namely, we suppose that the mixing matrix is
\begin{equation}
e_{ij} = \begin{cases}
\mu & i - j = \pm 1\\
1-2\mu & i = j\\
0 & \text{else}
\end{cases}
\end{equation}
where $d$ is the total number of modules. We consider $d = 3, 4, 6, 8, 10$, and we focus on values of $\mu < 1/3$, so that each module has more links to itself than to any other module. We also considered next-nearest-neighbor coupling in the ring of modules, and found very similar results; see Appendix \ref{sec:next-nearest-neighbor-appdx}. The results presented in this section were obtained from direct simulation of network dynamics, on networks of $n = 1.25\times 10^4$ nodes per module, with a fractional activation threshold $\theta = 0.4$, and averaged over 50 samples.

A summary of numerical results is given in Fig. \ref{fig:summary_fig_multiple_modules_m_1}, with the left column for max-degree and the right column for uniform seeding. First, we observe that activation of the entire network (i.e. the yellow area) is only achieved within our selected parameter space under maximum-degree seeding. In contrast, the greatest level of activation achieved in the uniform seeding scenario is for large $\rho_0$ and small $\mu$, and in these cases we find complete activation of the first module but negligible activation of other modules (see Fig. \ref{fig:detail_fig_multiple_modules_d_3_m_1}).

\begin{figure}
\includegraphics[width = \columnwidth]{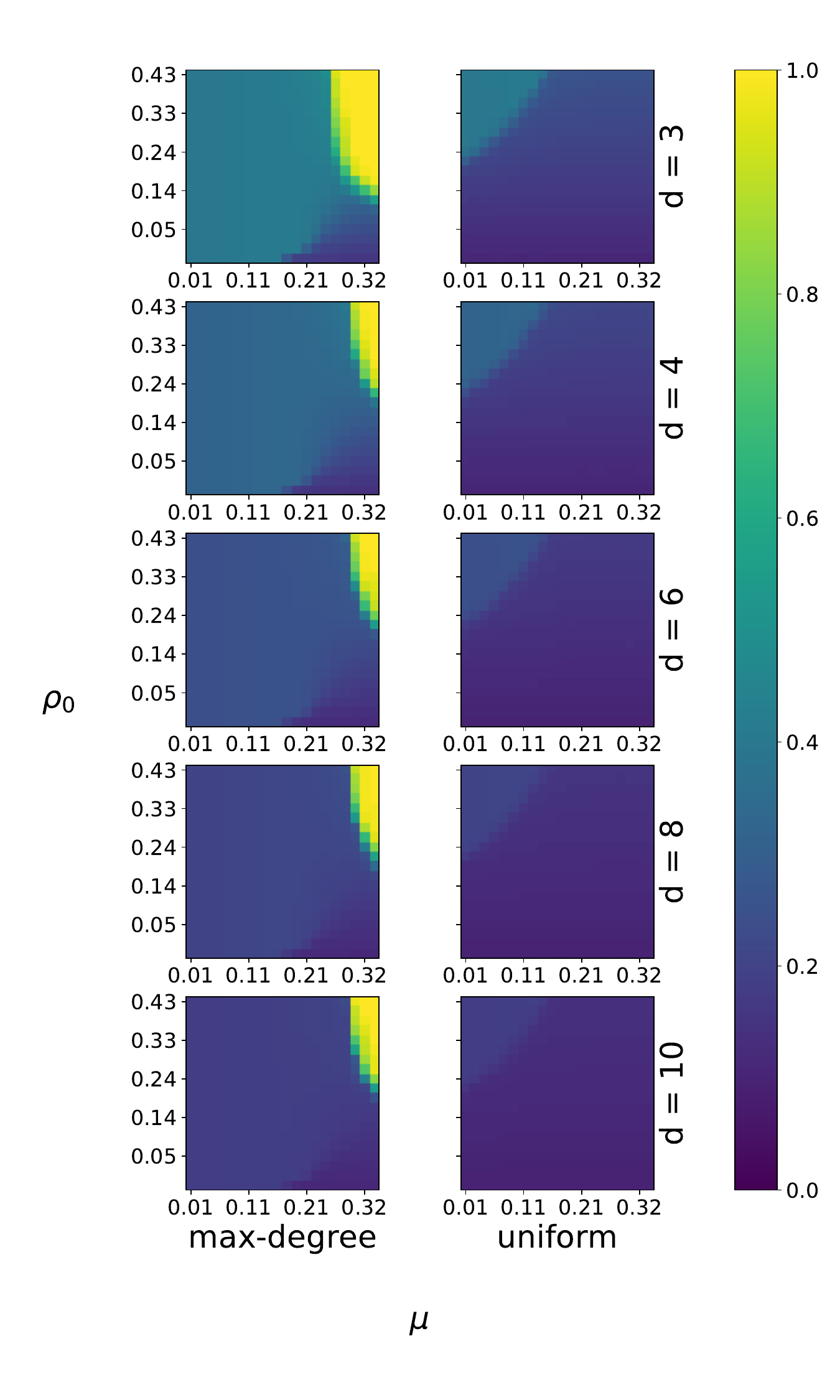}
\caption{Heatmap depicting total fraction of the ring-of-modules network eventually activated under maximum-degree (left) and uniform (right) seeding strategies. Heatmap axes are inter-module link density $\mu$ (horizontal) and fraction of links in the first module activated initially, $\rho_0$ (vertical). Each row corresponds to a different total number of modules, $d \in \{3, 4, 6, 8, 10\}$. As in Fig.\ref{fig:optimal_modularity_heatmaps}, color indicates the fraction of nodes in the whole network eventually activated. \label{fig:summary_fig_multiple_modules_m_1}}
\end{figure}

\begin{figure}
\includegraphics[width = \columnwidth]{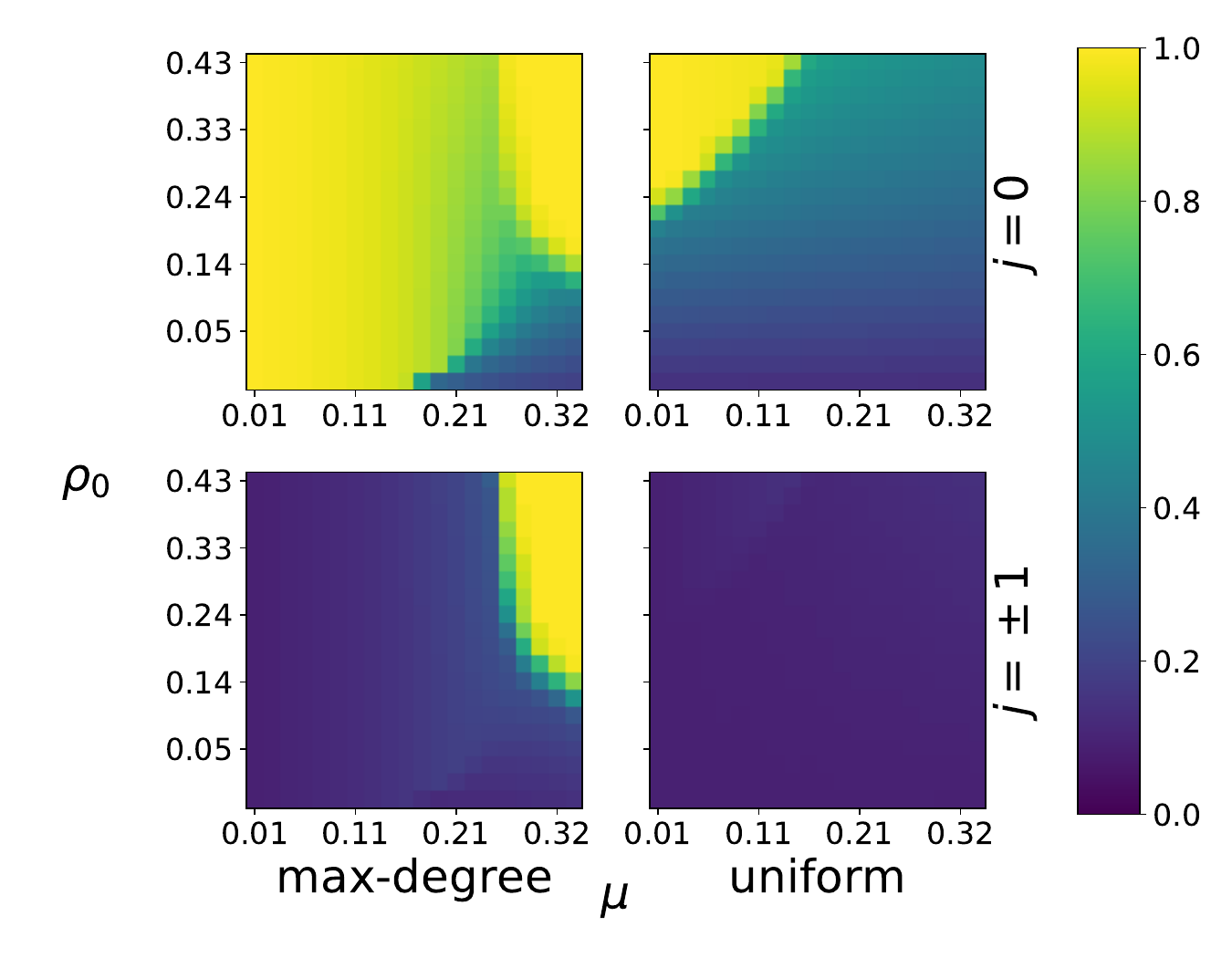}
\caption{Detail view of the fraction of nodes eventually activated in each module under maximum-degree (left) and uniform (right) seeding strategies. Heatmap axes are as in Fig. \ref{fig:summary_fig_multiple_modules_m_1}. Here, the number of modules is $3$, so we only show the eventual activated fraction for module $0$, the seed module, and one of the non-seed modules. Note that here color indicates the fraction of nodes \emph{within each module} that are eventually activated, rather than the fraction of the whole network, i.e. the yellow regions top row indicate that the seed module has become fully activated and indicate nothing about the other modules. \label{fig:detail_fig_multiple_modules_d_3_m_1}}
\end{figure}

These results are consistent with the existence of a critical level of inter-module link density $\mu_c$ required for complete activation of the network, as there is no case where the entire network is activated for $\mu\lesssim0.25$. However, this boundary is not sharp in the way that it is in the two-module case depicted in Fig. \ref{fig:optimal_modularity_heatmaps}, indicating that there are other limiting factors at play that prevent a cascade from reaching the entire network.

\section{Summary and Discussion}

In this work, we considered the dynamics of cascades on modular, degree heterogeneous networks subject to degree-dependent seeding.
Interestingly, we found that the ability of a local cascade to fully activate other modules is independent of the details of how it was initiated. For the cases we investigated, we find the critical fraction $\mu_c$ of links that must connect the two modules ranges between $0.20$ and $0.26$ depending on the degree distribution, suggesting that either preventing or facilitating global cascades by adjusting inter-module connectivity may be a feasible intervention. Moreover, the extent of the cascade can change abruptly as a function of $\mu$ near $\mu_c$, meaning that small interventions on the number of inter-module links can have an outsized impact on the outcome of a cascade process. {We have proven the existence of $\mu_c$ for networks with an arbitrary number of modules, given that all modules have the same degree distribution and that inter-module link density is the same for all module pairs (see Appendix \ref{sec:well-ordering-appendix}). We have also observed numerically the existence of $\mu_c$ in the case that the modules have different degree distributions, and in so doing have found that the value of $\mu_c$ is a function primarily of the degree distribution of the non-seed module.} The same tools we use here may be applied to investigate more general cases when the connectivity among modules is more heterogeneous.

We also observed that when a network is highly degree-heterogeneous, cascades initiated at the highest-degree nodes become global with relatively few seed nodes, provided $\mu\ge \mu_c$. In contrast, seeding nodes uniformly at random can require significantly more seeds to initiate a global cascade on the same network. Figure \ref{fig:optimal_modularity_heatmaps} demonstrates that targeted seeding can create a global cascade in a regime where random seeding cannot, and shows that it can be seen both in direct simulations of the network dynamics and in the approximate formulation (\ref{eq:rho_update_degree_weighted})--(\ref{eq:qbar_update_closed_degree_weighted}). Similar to results for percolation on random networks with power law degree distributions, degree heterogeneity enhances the network’s resilience to random failures but increases its vulnerability to targeted attack.

The above results were established both through direct simulation of network dynamics, and a theoretical analysis valid in the limit $N\to \infty$. We have demonstrated that it is possible to efficiently approximate the outcome of a cascade process on a network in the case that initial activation probability depends on degree.
Notably we demonstrate that using only one dynamical variable, namely the conditional exposure probability, is sufficient to describe cascading dynamics accurately even in the case that the initialization is degree-dependent. {In doing so we have shown that the true generality of the conceptual approach introduced by Gleeson is greater than that reported in Ref. \cite{Gleeson2008}, and have obtained important physical insight into the effect of degree-dependent seeding on cascade dynamics.}


Finally we discuss the implications of our results for real situations such as cascading extinctions and information diffusion. To the extent that our assumptions are appropriate (i.e. degree-corrected SBM network and fractional threshold dynamics), we have shown that high-degree nodes (species or actors) are able to initiate large local cascades (of extinction or information respectively). Yet, at the same time, there is a critical level of inter-module connectivity below which a global cascade is impossible, regardless of seeding (i.e. initially extinct species or source of information). The results presented here should extend to systems with different assumptions beyond fractional threshold dynamics and SBM networks.

\section*{Acknowledgements}
JS thanks Nicholas Roberts for helpful discussions. We gratefully acknowledge support from the US Army Research Office MURI Award No. W911NF-13-1-0340,  and the Minerva Initiative Award Grant No. W911NF-15-1-0502.

\appendix

{
\section{Glossary}
\begin{itemize}
	\item $d$ - number of modules
	\item $p^{(i)}_k$  - probability that a node in module $i$ has degree $k$ 
	\item $e = (e_{ij})$  - mixing matrix: $e_{ij}$ is the probability that a link connects module $i$ to module $j$
	\item $N$ - number of nodes
	\item $A = (A_{ij})$ - adjacency matrix: $A_{ij} = 1$ if node $i$ is connected to node $j$, and zero otherwise
	\item $u_i(t) \in \{0, 1\}$ - the state of node $i$ at time $t$
	\item $F^{(i)}(m,k)$  -  probability that an inactive node in module $i$ having degree $k$ and $m$ active neighbors will become active.
	\item $\rho^{(i)}_{0,k}$  -  probability that a node in module $i$ with degree $k$ is active initially
	\item $q^{(i)}_{n,k}$  -  probability that a node in module $i$ at tree level $n$ with degree $k$ is active, \emph{given that its parent is inactive}
	\item $\overline{q}^{(i)}_n$  - the \emph{conditional exposure probability}: the probability that level-$n$ child of an inactive module-$i$ parent is active
	\item $\rho^{(i)}_{n,k}$  -  probability that a node in module $i$ at tree level $n$ with degree $k$ is active, unconditionally.
	\item $\mu$  - when $d=2$, the fraction of edges between modules
	\item $z$  -  mean degree
	\item $p_\text{nest}$  -  parameter controlling degree heterogeneity, according to $p_k =  p_\text{nest} p^\text{pow}_k + (1-p_\text{nest})p^\text{poi}_k$
	\item $p^\text{pow}_k$ - power-law degree distribution, given by $p^\text{pow}_k = {1 \over \zeta(\gamma, \lambda)} (\lambda + k )^{-\gamma}$
	\item $\gamma$  -  slope of power law tail
	\item $\lambda$ - location parameter of power law distribution, chosen to ensure $\sum_k k p^\text{pow}_k = z$
	\item $\zeta(\gamma, \lambda)$  -  Hurwitz zeta function
	\item $p^\text{poi}_k$ - Poisson degree distribution, given by  $p^\text{poi}_k = {z^k \eu^{-z}}/{k!}$
\end{itemize}
}
\section{Stub-matching for network sampling}
\label{sec:stub-matching}

Here we describe the stub-matching procedure we use to sample from the network ensemble defined in Sec. \ref{sec:network_model}.

Specifically, for each module $i$, we produce a list of \emph{stubs}, $S_i$, where the node index $l$ appears in the list a number of times equal to its degree, which is sampled from the distribution $p^{(i)}_k$ independently for each node in module $i$. (Note: before the stub lists are generated, one should check that the sampled degree sequence is graphical \cite{Hakimi1963}).

Once the stub lists are generated, we compute the number of edges, $E_{ij}$, that should run between modules $i$ and $j$ by scaling the matrix entries $e_{ij}$ by the total number of edges. Since it's not guaranteed that the total number of links will be exactly compatible, we take 
\begin{equation}
E_{ij} = \min\left\{|S_i|\frac{e_{ij}}{\sum_{j'} e_{ij'}}, |S_j|\frac{e_{ij}}{\sum_{i'} e_{i'j}}\right\}.
\end{equation}
{The normalizing factors $\sum_{j'} e_{ij'}$ and $\sum_{i'}e_{i'j}$ ensure that $\sum_j E_{ij} \le |S_i|$ and $\sum_i E_{ij} \le |S_j|$, i.e. that there will be enough stubs emanating from each module to complete the desired number of edges.}

Next, we permute each stub list at random. For each distinct pair of modules $i$ and $j$, we take $E_{ij}$ stubs from each of the shuffled lists, connect the corresponding nodes pairwise, and remove them from the stub lists. When $i=j$, we take $E_{ii}$ stubs from the shuffled stub list $S_{i}$ and connect them with each other at random (discarding one stub if $E_{ii}$ is odd). If any multi-edges or self-loops exist, remove them.

The sloppiness of (i) possibly not using all stubs and (ii) removing multi-edges and self-loops at the end of the algorithm mean that for small $n$, this algorithm produces networks whose distribution differs slightly from the exact specifications that we took to define them. However, these errors are expected to be insignificant in the $n\to\infty$ limit \cite{Molloy1995}.

\section{Validation}
\label{sec:Validation}
Here we demonstrate that the theoretical approach outlined in Sec. \ref{sec:degree_targeted_seeding} constitutes a good approximation to the dynamics of cascades on networks. To do this, we first replicate the numerical experiment performed by Gleeson in \cite{Gleeson2008}. The experiment consists of a network of four modules, connected in a ring (cf. right side of Fig. \ref{fig:comparison_nested_graphic}). Their degree distributions are Poisson with mean 5.8, Poisson with mean 8, regular with degree 8, and regular with degree 8, respectively. The mixing matrix is
\begin{equation}
e = {1 \over 29.8}
\begin{bmatrix}
5.5& 0.15& 0.15& 0 \\
0.15& 7.7& 0& 0.15 \\
0.15& 0& 7.7 & 0.15 \\
0& 0.15& 0.15& 7.7 \\
\end{bmatrix}.
\end{equation}
The initial condition is that a randomly chosen 1\% of the nodes in the first module are active, and the threshold is taken to be $\theta = 0.18$. For these parameter settings, the cascade eventually takes over the whole network, but reaches each of the modules at different times, due to their differing internal characteristics (i.e. degree distributions) and the nature of the links between them (as encoded by the matrix $e$).

To see the effect of seeding the highest-degree nodes in the network, we recreate calculations for the same network as in \cite{Gleeson2008} using the same fraction of nodes initially activated (i.e. 1\% of the first module), selecting nodes via either a random strategy or a degree-targeted strategy.

For concreteness, we now explicitly construct $\rho^{(i)}_{0,k}$ for the case where we seed only the highest-degree nodes in the network. Say we have, for each module, a target fraction of nodes that we would like to activate initially, $\rho^{(i)}_{0,\text{tot}}\in [0,1]$ for $i=1,\dots,d$. We would like to choose $\rho^{(i)}_{0,k}$ such that the appropriate total number of nodes are activated, but only the nodes of highest possible degree are chosen.

To do this, we need to find, for each module $i$, a value $K^i$ such that the fraction of nodes in module $i$ having degree greater than or equal to $K^i$ is close to $\rho^{(i)}_{0,\text{tot}}$. Formally, let
\begin{equation}
K^i = \min\left\{K \, \left| \, \sum_{k> K} p^{(i)}_k \le \rho^{(i)}_{0,\text{tot}} \right.\right\}
\end{equation}
with the understanding that the minimum of an empty set is $+\infty$. Then, define
\begin{equation}
\rho^{(i)}_{0,k} = 
\begin{cases}
	1 & k > K^i\\
	\rho^{(i)}_{0,\text{tot}} - \sum_{k> K^i} p^{(i)}_k & k = K^i\\
	0 & k < K^i
\end{cases}
\label{eq:degree_weighted_rho_defn}
\end{equation}

We then construct the ODE defined in Eq. (\ref{eq:ODE_degree_weighted}), with threshold $\theta = 0.18$ and update rate $f=0.01$. Results are shown in Fig. \ref{fig:seeding_strategy_comparison_with_numerical}. Note excellent agreement between theory and experiment for the ``baseline'' case, i.e. seed probability independent of degree, as considered by Gleeson; the solid and dashed curves nearly perfectly overlap each other. This confirms that our generalization reduces to the known model when seed probability is independent of degree. For the degree-targeted case, we also see nearly perfect agreement between the ODE approximation and direct simulation of the network dynamics, indicating that our analytic formulation captures the intended dynamics.

\begin{figure}
\includegraphics[width=\columnwidth]{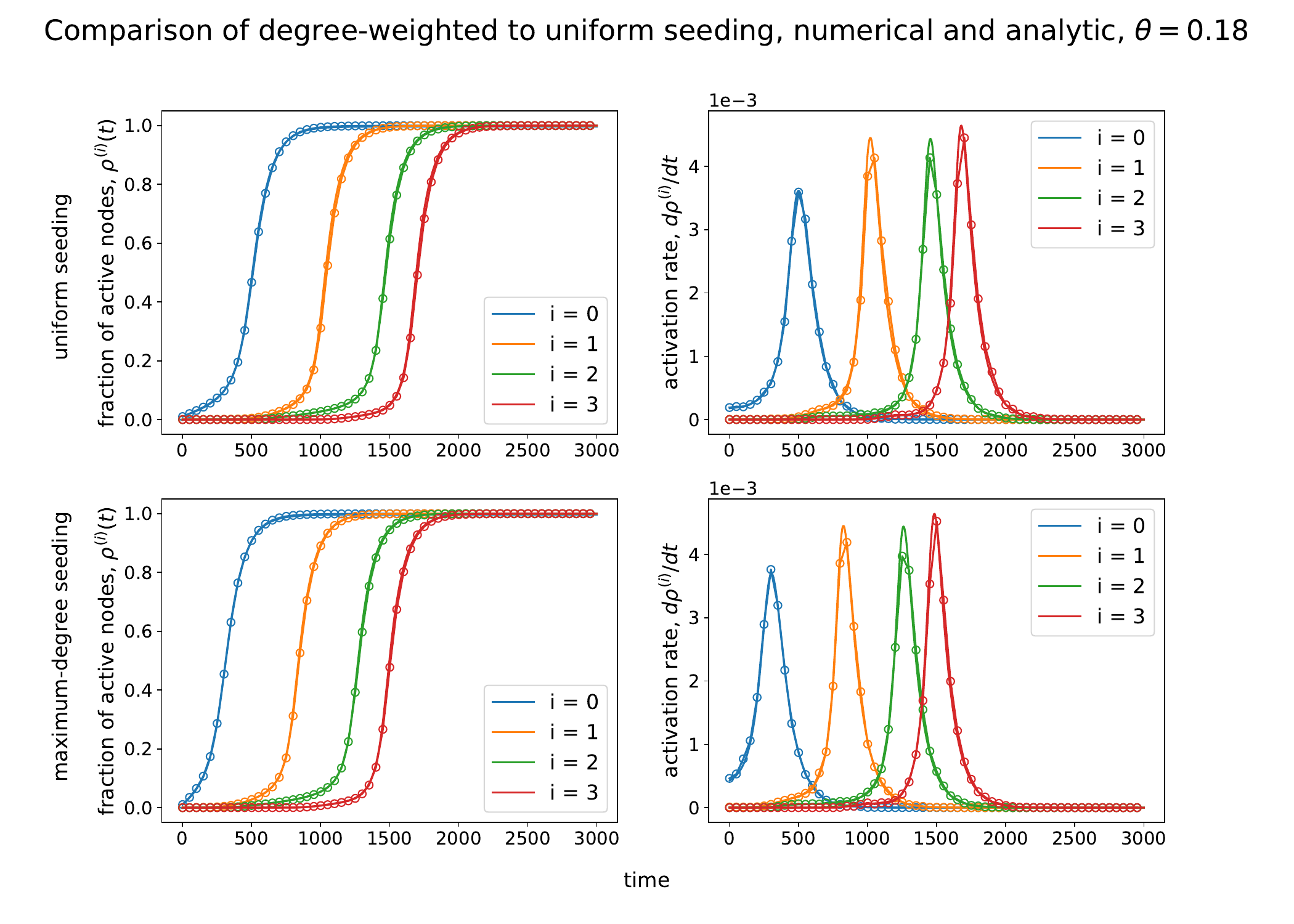}
\caption{Validating the ODEs obtained by the treelike approximation under the assumption of degree-dependent seeding. Parameters are identical to those used to create Fig. 1 in \cite{Gleeson2008}. Here we display both the total fraction of active nodes (left) and the rate of increase of the number of active nodes (right) in each module, for both uniform (top) and degree-targeted seeding (bottom). Solid curves are solutions of the ODE system Eq. (\ref{eq:ODE_degree_weighted}), and dashed curves are the corresponding quantities in a direct numerical simulation of threshold dynamics on a network of size $N = 5 \times 10^5$, averaged over 10 realizations of the network. \label{fig:seeding_strategy_comparison_with_numerical}}
\end{figure}

\begin{figure}
\includegraphics[width=0.69\columnwidth]{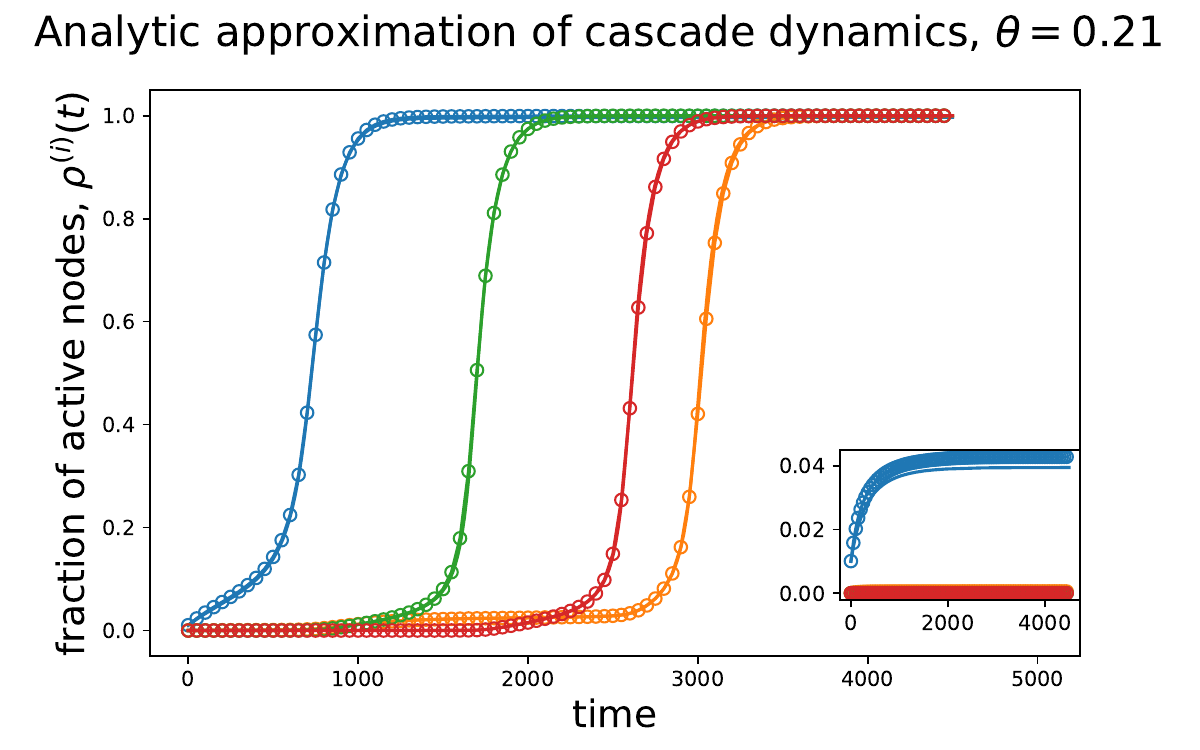}
\includegraphics[width=0.29\columnwidth]{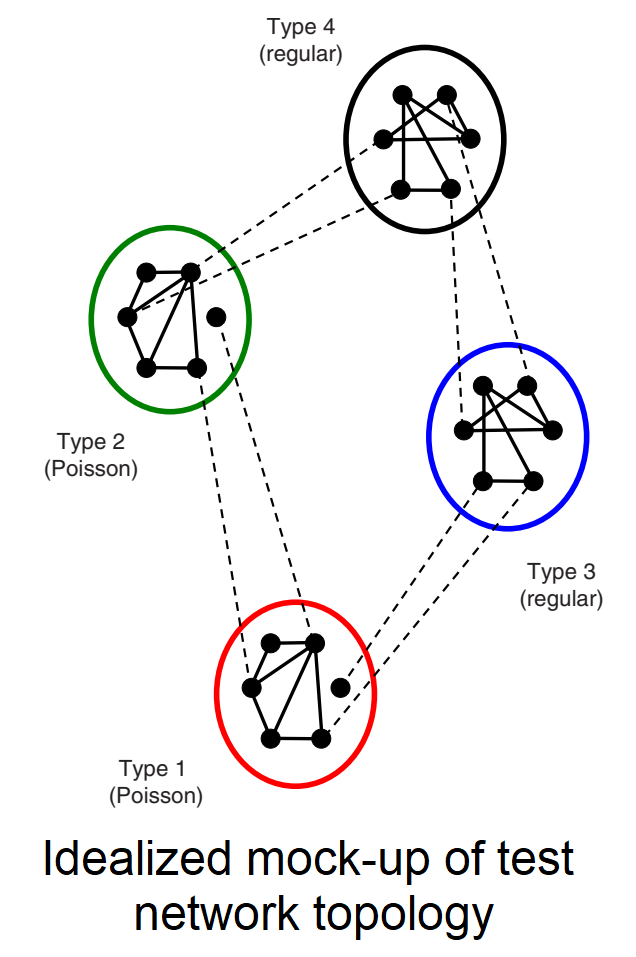}

\caption{Degree-targeted seeding creates a global cascade (main figure) in a regime where uniform seeding does not (inset), as demonstrated by both Eq. (\ref{eq:ODE_degree_weighted}) (solid curves) and dynamics on an actual network of size $5\times 10^5$ (open circles).
Network parameters are identical to those used in \cite{Gleeson2008}, and summarized by the right-hand diagram (figure from \cite{Gleeson2008}): degree distributions are Poisson with mean 5.8, Poisson with mean 8, regular with degree 8, and regular with degree 8.\label{fig:comparison_nested_graphic}}
\end{figure}

Next, to highlight the effect of maximum-degree seeding, we consider a slight modification of the above example. We take the same network as a above, but with threshold $\theta = 0.21$. As shown in Fig. \ref{fig:comparison_nested_graphic}, using this threshold value demonstrates that with a degree-targeted seeding strategy it is possible to excite a global cascade in a regime where a random seeding strategy would lead to only a tiny extent of spreading. This shows both the impact of targeted seeding and that the reduced model accurately approximates the true dynamics in both super- and sub-critical regimes.
%

\section{Proof that dynamics preserve partial order}
\label{sec:monotonicity-appendix}
Here we prove that the dynamics of the conditional exposure probability preserve the partial order $\le$ on $[0,1]^d$, a fact that will be useful in establishing subsequent bounds. {Namely, this fact helps us prove that if an initial condition localized to the first module would lead to complete activation of the whole network, then so would the initial condition in which all of the first module is active and none of the second is active.} To do this, we first prove a technical lemma.

\begin{lem} Let $k\in \mathbb{N}$, $m_c \in \{0, \dots, k-1\}$, and $x\in[0,1]$. Then
\begin{equation}
\frac{\partial}{\partial x} \left[\sum_{m=m_c}^{k-1} \binom{k-1}{m} x^m (1-x)^{k-1-m}\right] \ge 0
\end{equation}
\label{lem:monotonicity_of_binomial_tails}
\end{lem}
\begin{proof}
Let $X$ denote the binomial random variable with $k-1$ trials and probability $x$ of success for each trial. Then
\begin{equation}
\mathbb{P}(X\ge m_c) = \sum_{m=m_c}^{k-1} \binom{k-1}{m} x^m (1-x)^{k-1-m}.
\end{equation}

Observe first that $\mathbb{E}(X\vert X\ge m_c) \ge \mathbb{E}(X)$.
Next, differentiating $\mathbb{P}(X\ge m_c)$ with respect to $x$ we have \cite{3315179}
\begin{widetext}
\begin{align}
\frac{\partial \mathbb{P}(X\ge m_c)}{\partial x} &= \sum_{m=m_c}^{k-1} \binom{k-1}{m} \frac{\partial}{\partial x}\left( x^m (1-x)^{k-1-m} \right)\\
& = \sum_{m=m_c}^{k-1} \binom{k-1}{m} \left( x^m (1-x)^{k-1-m} \right) \left[\frac{m}{x(1-x)} - \frac{k-1}{1-x}\right]\\
& = \frac{1}{x(1-x)}\left[\sum_{m=m_c}^{k-1} \binom{k-1}{m} x^m (1-x)^{k-1-m}m - \mathbb{P}(X\ge m_c)(k-1)x\right]\\
& = \frac{\mathbb{P}(X\ge m_c)}{x(1-x)} \left[\mathbb{E}(X\vert X\ge m_c) - \mathbb{E}(X)\right] \ge 0
\end{align}
\end{widetext}
and the proof is complete.
\end{proof}

\begin{thm}
\label{thm:time_evolution_respects_partial_order}
Let $G\colon [0,1]^d \to [0,1]^d$ be defined as in Eq. (\ref{eq:qbar_update_closed_degree_weighted}), and let $\le$ be the partial order on $[0,1]^d$ defined by $x\le y \iff x_i \le y_i \forall i$. Then $x \le y\implies G(x) \le G(y)$.
\end{thm}

\begin{proof}
Let $x,y\in [0,1]^d$ such that $x\le y$. Then consider the difference $G(y) - G(x)$, which we argue is non-negative in each component. The $i^\text{th}$ component of $G(y)-G(x)$ is given by
\begin{align}
\left(G(y) - G(x)\right)_i 
&= \sum_j e_{ij} \sum_k k \frac{p^{(j)}_k}{z^{(j)}} \left[B(y_j;j,k) - B(x_j;j,k)\right]
\end{align}
where
\begin{equation}
B(x;j,k) \coloneqq \sum_m \binom{k-1}{m}x^m(1-x)^{k-1-m}F^{(j)}(m,k).
\label{eq:biom_sum_defn}
\end{equation}
Our aim is to show that $\partial B(x;j,k)/\partial x \ge 0$ so that $x_j \le y_j\implies B(x_j;j,k) \le B(y_j;j,k)$ for all $j,k$, rendering $(G(y) - G(x))_i$ a sum of nonnegative terms.

Since $F^{(j)}(m,k)$ is non-decreasing in $m$, we can write it as the partial sum of a non-negative sequence:
$
F^{(j)}(m,k) = \sum_{l=0}^{m} f^{(j)}(l,k)
$
where $f^{(j)}(l,k)\ge 0$. We then have, by rearranging the double sum,
\begin{align}
B(x;j,k) &= \sum_m \binom{k-1}{m}x^m(1-x)^{k-1-m}\sum_{l=0}^{m}f^{(j)}(l,k)\nonumber\\
& = \sum_{l=0}^{k-1} f^{(j)}(l,k) \sum_{m=l}^{k-1} \binom{k-1}{m}x^m(1-x)^{k-1-m}\nonumber \\
& = \sum_{l=0}^{k-1} f^{(j)}(l,k) \mathbb{P}(X\ge l). \nonumber
\end{align}
Finally,
\begin{equation}
\frac{\partial B(x;j,k)}{\partial x} = \sum_{l=0}^{k-1} f^{(j)}(l,k) \frac{\partial}{\partial x}\mathbb{P}(X\ge l)\nonumber
\end{equation}
and by Lemma \ref{lem:monotonicity_of_binomial_tails}, $\partial \mathbb{P}(X\ge l) / \partial x \ge 0$, so $\partial B(x;j,l) / \partial x \ge 0$, and we are done. Thus we have shown that if $x \le y$, then $G(x) \le  G(y)$.
\end{proof}

\section{Proof of well-ordering of activation}
\label{sec:well-ordering-appendix}
Now we prove a theorem in the setting considered in Sec. \ref{sec:optimal_modularity}. {Namely, we prove that if intra-module link density is greater than inter-module link density, and if the conditional exposure probability is initially largest in the first module, then it will be largest in the first module at all subsequent times. This fact is sufficient to establish the existence of a critical level of inter-module connectivity, $\mu_c$, such that no cascade can become global if seed nodes are only selected from one module and $\mu<\mu_c$.}

Denoting by $d$ the number of modules and by $\mu$ the link density between different modules,
we now prove that for $\mu<1/d$, $\overline{q}^{(1)}_{n}>\overline{q}^{(i)}_{n}$ implies $\overline{q}^{(1)}_{n+1}>\overline{q}^{(i)}_{n+1}$ for all $i\neq 1$.

\begin{thm}
\label{thm:first_module_always_dominates}
Let $G \colon [0,1]^d \to [0,1]^d$ be defined as in Eq. (\ref{eq:qbar_update_closed_degree_weighted}) with $p^{(i)}_k \equiv p_k$ for all $i,k$, $F^{(i)} \equiv F$ for all $i$, and
\begin{equation}
e_{ij} = \begin{cases}
\mu & i\neq j\\
1 - (d-1)\mu & i = j
\end{cases}
\end{equation}
with $\mu < 1/d$. Assume further that
\begin{equation}
\rho_{0,k}^{(i)} = \begin{cases}
\rho_{0,k} & i = 1\\
0 & \text{else}.
\end{cases}
\end{equation}
Then for every $i \in \{2, \dots, d\}$, $\overline{q}^{(1)} \ge \overline{q}^{(i)} \implies G(\overline{q})^{(1)} \ge G(\overline{q})^{(i)}$.
\end{thm}
\begin{proof}
Let $i \in \{2, \dots, d\}$ and let $\overline{q}\in [0,1]^d$ such that $\overline{q}^{(1)} \ge \overline{q}^{(i)}$. Because we have assumed $F^{(i)} \equiv F$ for all $i$, we suppress the dependence on module index $j$ in the binomial-like sum (\ref{eq:biom_sum_defn}) and instead write $B_k(x)$. Then
\begin{widetext}
\begin{align}
G(\overline{q})^{(1)} - G(\overline{q})^{(i)} &= \sum_j e_{1j} \left[\sum_k \frac{k}{z}p_k \left(\rho_{0,k}^{(j)} + (1-\rho_{0,k}^{(j)}) B_k(\overline{q}^{(j)})\right)\right] -\sum_j e_{ij} \left[\sum_k \frac{k}{z}p_k \left(\rho_{0,k}^{(j)} + (1-\rho_{0,k}^{(j)}) B_k(\overline{q}^{(j)})\right)\right]\nonumber\\
&= \sum_j (e_{1j} - e_{ij}) \left[\sum_k \frac{k}{z}p_k \left(\rho_{0,k}^{(j)} + (1-\rho_{0,k}^{(j)}) B_k(\overline{q}^{(j)})\right)\right]\nonumber
\end{align}
Due to the form of the mixing matrix, most terms in this expression vanish, except for those where $j\in \{1,i\}$. Thus we have
\begin{align}
G(\overline{q})^{(1)} - G(\overline{q})^{(i)} &= (1-d\mu)\Bigg[ \sum_k \frac{k}{z}p_k \bigg(\rho_{0,k} + (1-\rho_{0,k}) B_k(\overline{q}^{(1)})\bigg) - \sum_k \frac{k}{z}p_k B_k(\overline{q}^{(i)})\Bigg]\nonumber\\
&= (1-d\mu) \sum_k \frac{k}{z}p_k\Bigg[ \rho_{0,k}\left(1-B_k(\overline{q}^{(1)})\right) + \left(B_k(\overline{q}^{(1)}) - B_k(\overline{q}^{(i)}) \right)\Bigg]. \label{eq:line_with_the big-ass_term}
\end{align}
\end{widetext}

By Lemma \ref{lem:monotonicity_of_binomial_tails} and the fact that $\overline{q}^{(1)} \ge \overline{q}^{(i)}$, the term in the second set of parentheses in Eq. (\ref{eq:line_with_the big-ass_term}) is non-negative. Since by assumption $\mu< 1/d$, $1-d\mu>0$; $kp_k / z\ge 0$; $\rho_{0,k}\ge 0$; and $B_k(\overline{q}^{(i)}) \le 1$. Therefore $G(\overline{q})^{(1)}\ge G(\overline{q})^{(i)}$.
\end{proof}

Applying Theorem \ref{thm:first_module_always_dominates} in the case $d=2$ gives precisely the statement in Sec. \ref{sec:optimal_modularity}.

\section{Next-nearest-neighbor ring of modules}
\label{sec:next-nearest-neighbor-appdx}
Here we expand on the experiments performed on the multiple-module system described in Sec. \ref{sec:general_network_structures}. Namely, Fig. \ref{fig:summary_fig_multiple_modules_m_2} shows a summary of results obtained in the case that modules share links with not only their nearest-neighbor modules, but also next-nearest-neighbor modules. In this case, the mixing matrix is
\begin{equation}
e_{ij} = \begin{cases}
\mu & 0 < |i - j| \mod d \le 2\\
1-4\mu & i = j\\
0 & \text{else}
\end{cases}
\end{equation}
and now we restrict $\mu$ to the range $0 < \mu < 1/5$, so that we still have the condition that each module has more links with itself than with any other module. Notably, there are no cases where activation of the entire network is obtained. This is due to the fact that the constraint $\mu<1/5$ is too restrictive to allow cascades to spread from one module to the rest of the network. Nonetheless, the trends that are visible in these heatmaps are consistent with the $\mu<1/5$ portion of those depicted in Fig. \ref{fig:summary_fig_multiple_modules_m_1}, indicating that there is not a very significant difference in the dynamics between nearest-neighbor and next-nearest-neighbor coupling in a ring of modules.

\begin{figure}
\includegraphics[width=\columnwidth]{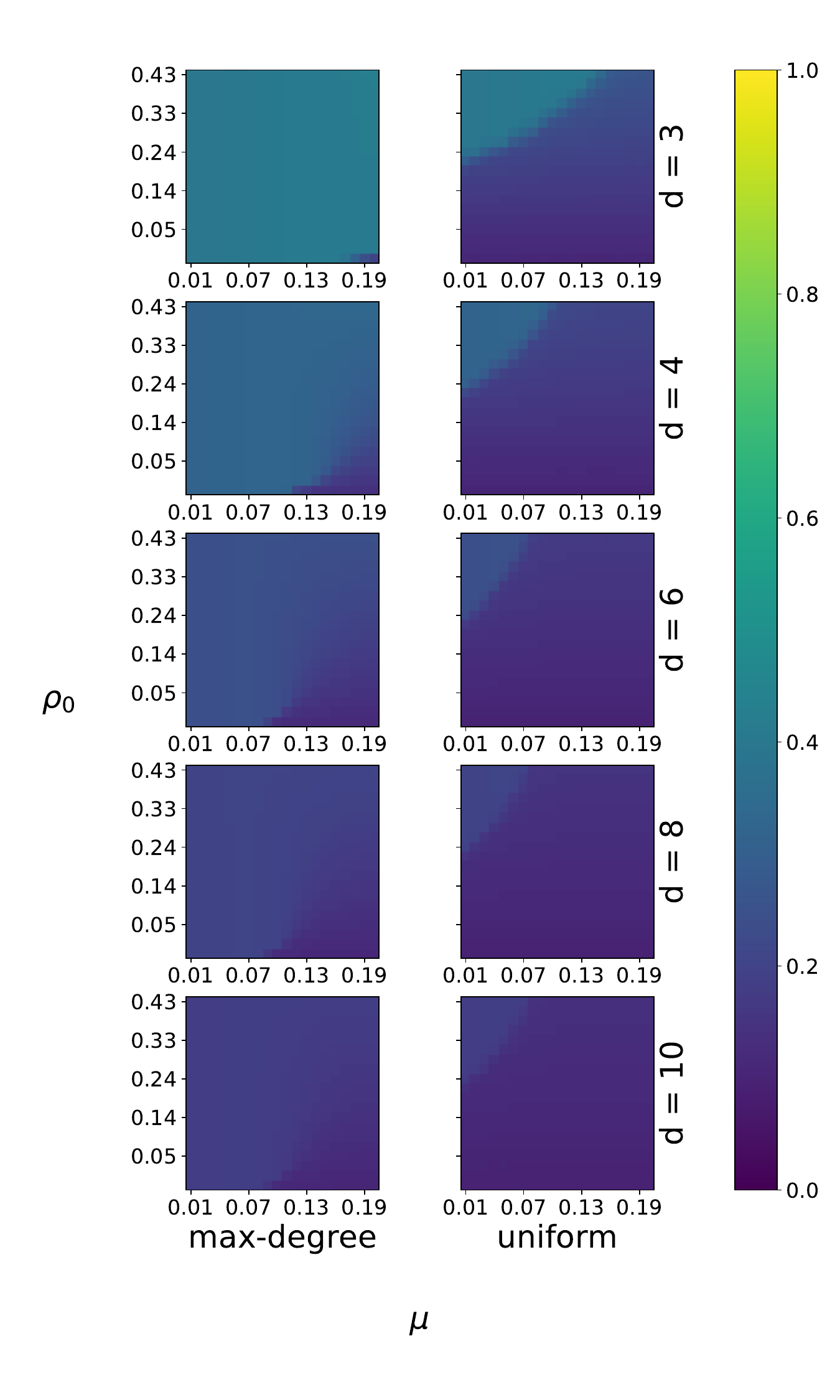}
\caption{As in Fig. \ref{fig:summary_fig_multiple_modules_m_1} but for the case that modules share links with both nearest and next-nearest neighbors in the ring of modules. Note that the range of $\mu$ considered here is smaller than that considered in the nearest-neighbor case. \label{fig:summary_fig_multiple_modules_m_2}}
\end{figure}

\bibliography{}

\end{document}